\documentclass[letterpaper,11pt]{article}
\usepackage{tgpagella}
\usepackage[utf8]{inputenc}
\usepackage{graphicx,fullpage,paralist}
\usepackage{amsmath, amssymb, amsthm}
\usepackage{comment,hyperref}
\usepackage{thmtools}
\usepackage{tikz}
\usepackage{pgfplots}
\usepackage{varwidth}
\pgfplotsset{compat=1.10}
\usepgfplotslibrary{fillbetween}
\usetikzlibrary{backgrounds}
\usetikzlibrary{patterns}
\usepackage{pifont}

\usepackage{xmpmulti}
\usepackage{color}
\usepackage{enumitem}

\newcommand{\calR}{\mathcal{R}}

\newcommand{\CC}{\mathcal{C}}
\newcommand{\I}{{\bf{I}}}

\newcommand{\hook}{\text{hook}}
\newcommand{\sched}{\textbf{sched}}

\newcommand{\RR}{\mathbb{R}}
\newcommand{\QQ}{\mathbb{Q}}
\newcommand{\NN}{\mathbb{N}}
\newcommand{\ZZ}{\mathbb{Z}}

\newcommand{\ass}{\text{assign}}
\newcommand{\clp}{\text{clp}}
\newcommand{\lex}{\text{lex}}
\newcommand{\config}{\text{conf}}
\newcommand{\conf}{\text{conf}}

\newcommand{\sos}{\text{SoS}}
\newcommand{\SA}{\text{SA}}
\newcommand{\LSP}{\text{LS}_+}
\newcommand{\largo}{\text{long}}
\newcommand{\corto}{\text{short}}

\newcommand{\pseudo}{\widetilde{\mathbb{E}}}
\newcommand{\row}{\text{row}}
\newcommand{\rowspace}{{\textbf{W}}_{\tau_{\lambda}}}
\newcommand{\rowgroup}{{\textbf{R}}_{\tau_{\lambda}}}
\newcommand{\tail}{\text{tail}}

\newcommand{\sym}{\text{sym}}

\newtheorem{theorem}{Theorem}
\newtheorem{lemma}{Lemma}

\newtheorem{definition}{Definition}
\newtheorem{example}{Example}
\newtheorem{proposition}{Proposition}
\newtheorem{claim}{Claim}

\usepackage{multirow}
\usepackage[noend]{algpseudocode}
\usepackage{algorithm,algorithmicx}

\newlength{\algofontsize}
\global\long\def\nrounds{(1/\varepsilon)^{2/\varepsilon^{2}}}%
\global\long\def\pnrounds{4/\varepsilon^{5}}%
\global\long\def\assign{\mathrm{assign}}%
\global\long\def\conf{\mathrm{conf}}%
\global\long\def\order{\mathrm{order}}%
\global\long\def\Jl{J_{\mathrm{long}}}%
\global\long\def\Jll{J_{\mathrm{long},\ell}}%
\global\long\def\N{\mathbb{N}}%

\usepackage[textsize=tiny,textwidth=2cm]{todonotes}

\usepackage{xcolor}
\hypersetup{
	colorlinks,
	linkcolor={red!50!black},
	citecolor={blue!50!black},
	urlcolor={blue!80!black}
}

\begin{document}
	
	\algrenewcommand\algorithmicrequire{\textbf{Input:}}
	\algrenewcommand\algorithmicensure{\textbf{Output:}}
	
	\title{\vspace{0em} Breaking symmetries to rescue Sum of Squares\\ in the case of makespan scheduling\thanks{This work has been partially funded by Fondecyt Projects Nr. 1170223 and 1181527.}}
	
	\author{
	Victor Verdugo\thanks{Department of Mathematics, London School of Economics and Political Science.  {\tt v.verdugo@lse.ac.uk}}
	\thanks{Institute of Engineering Sciences, Universidad de O'Higgins. {\tt victor.verdugo@uoh.cl}}
	\and
	Jos\'{e} Verschae\thanks{Institute of Mathematical and Computational Engineering, Pontificia Universidad Católica. {\tt jverschae@uc.cl}}
	\and
	Andreas Wiese\thanks{Department of Industrial Engineering, Universidad de Chile. {\tt awiese@dii.uchile.cl}}
	}
\date{\vspace{-1em}}
\maketitle
\begin{abstract}
The Sum of Squares (\sos{}) hierarchy gives an automatized technique to create a family of increasingly tight convex relaxations for binary programs. There are several problems for which a constant number of rounds of this hierarchy give integrality gaps matching the best known approximation algorithms. 
For many other problems, however, ad-hoc techniques give better approximation ratios than \sos{} in the worst case, as shown by corresponding lower bound instances.
Notably, in many cases these instances are invariant under the action of a large permutation group. 
This yields the question how 
symmetries in a formulation degrade the performance of the relaxation obtained by the \sos{} hierarchy. 
In this paper, we study this for the case of the minimum makespan problem on identical machines. 
Our first result is to show that $\Omega(n)$ rounds of \sos{} applied over the \emph{configuration linear program} yields an integrality gap of at least $1.0009$, where $n$ is the number of jobs. 
This improves on the recent work by Kurpisz et al.~\cite[Math. Prog. 2018]{KMMMVW18} that shows an analogous result for the weaker LS$_+$ and SA hierarchies. Our result 
is based on tools from representation theory of symmetric groups. 
Then, we consider the weaker \emph{assignment linear program} and add a well chosen set of symmetry breaking inequalities that removes a subset of the machine permutation symmetries. 
We show that applying $2^{\tilde{O}(1/\varepsilon^2)}$ rounds of the $\SA$ hierarchy to this stronger linear program reduces the integrality gap to $1+\varepsilon$,
which yields a linear programming based polynomial time approximation scheme.
Our results suggest that for this classical problem, symmetries were the main barrier preventing the \sos{}/$\SA$ hierarchies to give relaxations of polynomial complexity 
with an integrality gap of~$1+\varepsilon$. 
We leave as an open question whether this phenomenon occurs for other symmetric problems.

\end{abstract}

\thispagestyle{empty}
\newpage

\tableofcontents

\section{Introduction}

The lift-and-project methods are powerful techniques for deriving convex relaxations of integer programs. The lift-and-project \emph{hierarchies}, such as Sherali-Adams (SA), Lovász-Schrijver (LS), and Sum of Squares (\sos{}), are systematic methods for obtaining a family of increasingly
tight relaxations, parameterized by the number of \emph{rounds} of the hierarchy. For all of them, applying $r$ rounds on a formulation with $n$ variables yields a convex relaxation with $n^{O(r)}$ variables in the lifted space. Taking $r=n$ rounds gives an exact description of the integer hull~\cite{Las01}, at the cost
of having an exponential number of variables. On the other hand, taking $r=O(1)$ rounds yields a description with only a polynomial number of variables. Arguably, it is not well understood for which problems these hierarchies, with a constant number of rounds, yield relaxations that match the respective best possible approximation algorithm. Indeed, there are some positive results, but there are also many other strong negative results for algorithmically easy problems. 
These lower bounds show a natural limitation on the power of hierarchies as one-fits-all techniques. Quite remarkably, the instances used for obtaining lower bounds often have a very symmetric structure~\cite{laurent2003lower,grigoriev2001,potechin17,KMMMVW18,RSS18}, which suggests a connection between the tightness of the relaxation
given by these hierarchies and symmetries. The primary purpose of this article is to study this connection for a specific relevant problem, namely, the minimum \emph{makespan scheduling on identical machines}.

\paragraph{Minimum makespan scheduling.}

\noindent This is one of the first problems considered under the lens
of approximation algorithms~\cite{graham66}, and it has been studied
extensively. The input of the problem consists of a set $J$ of $n$
jobs, each having an integral processing time $p_{j}>0$, and a set
$[m]=\{1,\ldots,m\}$ of $m$ identical machines. Given an assignment $\sigma:J\rightarrow [m]$,
the load of a machine $i$ is the total processing time of jobs assigned
to $i$, that is, $\sum_{j\in\sigma^{-1}(i)}p_{j}$. The objective
is to find an assignment of jobs to machines that minimizes the \emph{makespan},
that is, the maximum load. The problem is strongly NP-hard and admits
several \textit{polynomial-time approximation schemes} (PTAS) based on different
techniques, such as dynamic programming, integer programming on fixed dimension,
and integer programming under a constant number of constraints~\cite{HS87,AAWY97,AAW98,hochbaum96,jansen10,jansen16,EisenbrandW18}.

\paragraph{Integrality gaps.}

\noindent The minimum makespan problem has two natural linear relaxations,
which have been extensively studied in the literature. The \emph{assignment
linear program} uses binary variables $x_{ij}$ which indicates whether
a job $j$ is assigned to a machine $i$, for each $i\in [m],j\in J$.
The stronger \emph{configuration linear program} uses a variable $y_{iC}$
for each machine $i$ and multiset of processing times $C$, which
indicates whether $C$ is the multiset of processing times of jobs
assigned to $i$. Kurpisz et al.~\cite{KMMMVW18} showed that the
configuration linear program has an integrality gap of at least $1024/1023\approx1.0009$
even after $\Omega(n)$ of rounds of the $\LSP$ or $\SA$ hierarchies.
On the other hand, Kurpisz et al.~\cite{KMMMVW18}
leave open whether the \sos{} hierarchy applied to the configuration
linear program has an integrality gap of $1+\varepsilon$ after applying
a number of rounds that depends only on the constant $\varepsilon>0$,
i.e., $O_{\varepsilon}(1)$ rounds. Our first main contribution is
a negative answer to this question. 
\begin{theorem}\label{thm:negative} Consider the minimum makespan
problem on identical machines. For each $n\in\mathbb{N}$, there exists
an instance with $n$ jobs such that, after applying $\Omega(n)$
rounds of the $\sos$ hierarchy over the configuration linear program,
the obtained semidefinite relaxation has an integrality gap of at
least $1.0009$.
\end{theorem}

Naturally, since the configuration linear program is stronger than the assignment linear program, our result holds if we apply $\Omega(n)$ rounds of $\sos$ over the assignment linear program. The proof of the lower bound relies on tools from representation theory of symmetric groups over polynomials rings,and it is inspired on the
recent work by Raymond et al.~\cite{RSST18} for symmetric sums of squares in hypercubes. It is based on constructing high-degree \textit{pseudoexpectations} on the one hand, and by obtaining symmetry-reduced decompositions of the polynomial ideal defined by the configuration linear program, on the other hand. The machinery from representation theory allows to restrict attention to invariant polynomials, and we combine this with a strong \textit{pseudoindependence} result for
a well chosen polynomial spanning set. Our analysis is also connected to the work of Razborov on flag algebras and graph densities, and we believe it can be of independent interest for analyzing lower bounds in the context of $\sos$ in presence of symmetries~\cite{razborov07,razborov10,RSS18}.

\paragraph{Symmetries and Hierarchies.}

Given the relation between hierarchies and symmetries above, it is natural to explore whether symmetry handling techniques might help to overcome the limitation given by Theorem~\ref{thm:negative}. 
A natural source of symmetry of the problem comes from the fact that the machines are identical: Given a schedule, we obtain another schedule with the same makespan by permuting the machines.
The same symmetries are encountered in the assignment and configuration linear programs, namely, if $\sigma:[m]\rightarrow [m]$ is a permutation and $(x_{ij})$ is a feasible solution to the assignment linear program then $(x_{\sigma(i)j})$ is also feasible. The same holds for solutions $(y_{iC})$ and $(y_{\sigma(i)C})$ for the configuration-LP. In other words, these linear programs are \emph{invariant} under the action of the symmetric group on the set of machines. 
The question we study is the following: \textit{Is it possible to obtain a polynomial size linear or semidefinite program with an integrality gap of at most $1+\varepsilon$ that is not invariant under the machine symmetries?} We aim to understand if these symmetries deteriorate the quality of the relaxations obtained from the \sos{} or SA hierarchies. This time, we provide a positive answer.

\begin{theorem}\label{thm:positive} Consider the problem of scheduling
identical machines to minimize the makespan. After adding linearly
many inequalities to the assignment linear program (for breaking symmetries),
$2^{\tilde{O}(1/\varepsilon^{2})}$
rounds of the \SA{} hierarchy yield a linear program with an integrality
gap of at most $1+\varepsilon$, for any $\varepsilon>0$. 
\end{theorem}

Notice that the same result is obtained by applying the SoS hierarchy instead of SA. The proof of Theorem~\ref{thm:positive} is based on introducing a formulation that \textit{breaks} the symmetries in the assignment program by adding new constraints. The symmetry breaking constraints enforce that any feasible integral solution of the formulation respects a \textit{lexicographic}
order over the machine configurations. 
We show how to exploit this to obtain a polynomial time approximation scheme (PTAS) based on the $\SA$ hierarchy. 
Additionally, we show that by adding a polynomial number of new constraints, we can obtain a faster approximation scheme, such
that $\text{poly}(1/\varepsilon)$ rounds of $\SA$ suffice. The extra
constraints correspond to symmetry breaking inequalities for a modified
instance with rounded job sizes. In particular, the added
constraints are not necessarily valid for the original formulation
(which considers the original job sizes). However, we can show that
increasing the optimal makespan by a factor $1+\varepsilon$ maintains
the feasibility of at least one integral solution. Thus, by breaking
more symmetries, we make it easier for the hierarchies to produce good
relaxations.
We remark that the framework we use for the minimum makespan problem can be, in principle, studied in other settings where symmetries are present in standard integer programming relaxations. 
This strategy opens the possibility of analyzing the effect of applying symmetry breaking techniques and hierarchies in order to generate strong linear or semidefinite relaxations.  

\subsection{Related work}

\paragraph{Upper bounds.}

The first application of semidefinite programming in the
context of approximation algorithms is due to Goemans and Williamson
for the Max-Cut problem~\cite{GW95}. Of particular interest to our work is the
SoS based approximation scheme by Karlin et al. to the Max-Knapsack
problem~\cite{KMN11}. They use a structural \textit{decomposition
theorem} satisfied by the SoS hierarchy. 
For a constant number of machines, Levey and Rothvoss design an approximation
scheme with a sub-exponential number of rounds in the weaker SA hierarchy~\cite{LeveyR15},
which is improved to a quasi-PTAS by Garg~\cite{Garg18}. The SoS method
has received a lot of attention for high-dimensional problems. Among
them we find matrix and tensor completion~\cite{barak2016noisy,potechin2017exact},
tensor decomposition~\cite{ma2016polynomial} and clustering~\cite{kothari2018robust,RSS18}.

\paragraph{Lower bounds.}

The first lower bound obtained in the context of \textit{positivstellensatz
certificates} is by Grigoriev~\cite{grigoriev2001}, showing the
necessity of a linear number of SoS rounds to refute an easy Knapsack
instance. Similar results are obtained by Laurent~\cite{laurent2003lower}
for Max-Cut and by Kurpisz et al.~\cite{kurpisz2016hardest} for unconstrained polynomial optimization.
The same authors show that for a certain polynomial-time single machine
scheduling problem, the SoS hierarchy exhibits an unbounded integrality
gap even in a high-degree regime~\cite{kurpisz2016hardest,kurpisz2017unbounded}.
Remarkable are the work of Grigoriev~\cite{grigoriev2001linear}
and Schoenebeck~\cite{schoenebeck2008linear} exhibiting the difficulty
for SoS to certify the insatisfiability of random 3-SAT instances
in subexponential time, 
and recently there have been efforts on unifying frameworks to show
lower bounds on random CSPs~\cite{BCK15,KMOW17,KOS18}. For estimation
and detection problems, lower bounds have been shown for the planted
clique problem, $k$-densest subgraph and tensor PCA, among others~\cite{HKPRSS16,barak2016nearly}.

\paragraph{Invariant Sum of Squares.}

Gatermann and Parrilo study how to obtain reduced sums of squares
certificates of non-negativity when the polynomial is invariant under
the action of a group, using tools from representation theory~\cite{GP04}.
Raymond et al.~\cite{RSST18} develop on the Gatermann and Parrilo method to construct
symmetry-reduced sum of squares certificates for polynomials over
$k$-subset hypercubes. Furthermore, the authors make
an interesting connection with the Razborov method and flag algebras~\cite{razborov07,razborov10}.
Blekherman et al.~\cite{blekherman2016sums} and Laurent~\cite{laurent2007semidefinite} provide degree bounds on rational representations
for certificates over the hypercube, recovering as a corollary known
lower bounds for combinatorial optimization problems like Max-Cut.
Kurpisz et al.~\cite{klm16} provide a method for proving SoS lower bounds when
the formulations exhibits a high degree of symmetry.

\section{Preliminaries: Sum of Squares (SoS) and Pseudoexpectations}
\label{sec:prelim}

In what follows we denote by $\RR[x]$ the ring of polynomials with real coefficients. 
Binary integer programming belongs to a larger class of problems in {\it polynomial optimization}, where the constraints are defined by polynomials in the variables indeterminates.
More specifically, consider the set ${\cal K}$ of feasible solutions to the polinomial optimization program defined by
\begin{alignat}{2}
	g_i(x)               	&\ge 0    	&&\quad\text{for all}\; i\in M,\label{eq:semi1}\\
	h_j(x)   	&=0    	&&\quad\text{for all}\; j \in J,\label{eq:semi2}\\
    x_e^2-x_e      		&= 0  	&&\quad\text{for all}\; e\in E\label{eq:semi3}.
\end{alignat}
where $g_i,h_{j}\in \RR[x]$ for all $i\in M$ and for all $j\in J$.
In particular, for binary integer programming the equality and inequality constraints are affine functions.\\

\noindent{\it Ideals, quotients and square-free polynomials.} In what follows we give a minimal introduction to the algebraic elements for polynomial optimization, for a comprehensive treatment see~\cite{cox_ideals_2007}. We denote by ${\I}_E$ the ideal of polynomials in $\RR[x]$ generated by $\{x_e^2-x_e:e\in E\}$, and let $\RR[x]/{\I}_E$ be the quotient ring of polynomials with respect to the vanishing ideal ${\I}_E$.
That is, $f,g\in \RR[x]$ are in the same equivalence class of the quotient ring if $f-g\in {\I}_E$, that we denote $f\equiv g\mod {\I}_E$. Alternatively, $f\equiv g\mod {\I}_E$ if and only if the polynomials evaluate to the same values on the vertices of the hypercube, that is, $f(x)=g(x)$ for all $x\in \{0,1\}^E$.
Observe that the equivalence classes in the quotient ring are in bijection with the {\it square-free} polynomials in $\RR[x]$, that is, polynomials where no variable appears squared. In what follows we identify elements of $\RR[x]/{\I}_E$ in this way, that is, for $p\in \RR[x]$ we denote by $\overline{p}$ the unique square-free representation of $p$, which can be obtained as the result of applying the polynomial division algorithm by the Gr\"{o}bner basis $\{x_e^2-x_e:e\in E\}$.
Given $S\subseteq E$, we denote by $x_S$ the square-free monomial that is obtained from the product of the variables indexed by the elements in $S$, that is,
$x_S=\prod_{e\in S}x_e.$
The degree of a polynomial $f\in \RR[x]/{\I}_E$ is denoted by $\deg(f)$. 
We say that $f$ is a {\it sum of squares} polynomial, for short $\sos$, if there exist polynomials $\{s_{\alpha}\}_{\alpha\in \mathcal{A}}$ for a finite family $\mathcal{A}$ in the quotient ring such that $f\equiv \sum_{\alpha \in \cal{A}}s_{\alpha}^2 \mod {\I_E}$.\\

\noindent{\it Certificates and $\sos$ method.} The question of certifying the emptiness of $\mathcal{K}$ is hard in general but sometimes it is possible to find {\it simple} certificates.
We say that there exists a degree-$\ell$ $\sos$ certificate of infeasibility for $\cal K$ if there exist $\sos$ polynomials $s_0$ and $\{s_i\}_{i\in M}$, and polynomials $\{r_j\}_{j\in J}$
such that
\begin{equation}
\label{eq:psatz}
-1\equiv s_0+\sum_{i\in M}s_ig_i+\sum_{j\in J}r_jh_j \mod {\I}_E,
\end{equation}
and the degree of every polynomial in the right hand side is at most $\ell$.
Observe that if $\cal K$ is non-empty, then the right hand side is guaranteed to be non negative for at least one assignment of $x$ in $\{0,1\}^E$, which contradicts  the equality above.  
In the case of binary integer programming, if $\cal K$ is empty there exists a degree-$\ell$ $\sos$ certificate, for some $\ell\le |E|$~\cite{lau09,Parr03}.
The $\sos$ algorithm iteratively checks the existence of a $\sos$ certificate, parameterized in the degree, and each step of the algorithm is called a {\it round}.
Since $|E|$ is an upper bound on the certificate degree, the method is guaranteed to terminate~\cite{Parr03,lau09}.
Furthermore, the existence of a degree-$\ell$ $\sos$ certificate can be decided by solving a semidefinite program.
This approach can be seen as the dual of the hierarchy proposed by Lasserre, which has been studied extensively in the optimization and algorithms community~\cite{Las01,Lau03,TR13,CT12}.\\

\noindent{\it Pseudoexpectations.} 
To determine the existence of a $\sos$ certificate one solves a semidefinite program, and the solutions of this program determine the coefficients of elements in the dual space of linear operators. 
We say that a linear functional $\widetilde{\mathbb{E}}:\RR[x]/\I_E\to \RR$ is a degree-$\ell$ $\sos$ {\it pseudoexpectation} for the polynomial system (\ref{eq:semi1})-(\ref{eq:semi3}), if it satisfies the following properties:
\begin{enumerate}[label=\text{(SoS.\arabic*)},align=right,leftmargin=1.8\labelwidth]
	\item \label{prop:normal} $\pseudo(1)=1$,
	\item \label{prop:sos1} $\pseudo(\overline{f^2})\ge 0$ for all $f\in \RR[x]/{\I}_E$ with $\deg(\overline{f^2})\le \ell$, 
	\item \label{prop:sos2} $\pseudo(\overline{f^2g_i})\ge 0$ for all $i\in M$, for all $f\in \RR[x]/{\I}_E$ with $\deg(\overline{f^2g_i})\le \ell$,
	\item \label{prop:SA} $\pseudo(\overline{fh_j})=0$ for all $j\in J$, for all $f\in \RR[x]/{\I}_E$ with $\deg(\overline{fh_j})\le \ell$.
\end{enumerate}  
In what follows, every time we evaluate a polynomial in the pseudoexpectation we are doing it over the square-free representation. We omit the bar notation for simplicity. The next lemma shows that there is a duality relation between degree-$\ell$ $\sos$ pseudoexpectation and $\sos$ certificates of infeasibility of the same degree.

\begin{lemma}\label{lem:infeasible-pseudo}
Suppose that $\cal K$ is empty. 
If there exists a degree-$\ell$ $\sos$ pseudoexpectation then there is no degree-$\ell$ $\sos$ certificate of infeasibility.
\end{lemma}
%
The proof of this lemma is a simple check, see also~\cite{mastrolilli2017high}. 
The minimum value of $\ell$ for which there exists a $\sos$ certificate of infeasibility tells how hard is determining the emptiness of $\mathcal{K}$ for the $\sos$ method. 
Lemma \ref{lem:infeasible-pseudo} provides a way of finding lower bounds on the minimum degree of a certificate, which we use in Section~\ref{sec:lb} for the minimum makespan problem.
There are many examples of problems that are extremely easy to certificate for humans, but not for the $\sos$ method.
For example, given a positive $k\in \QQ\setminus \ZZ$, consider the program $\sum_{e\in E}x_e=k$ and $x_e^2-x_e=0$ for all $e\in E$.
This problem is clearly infeasible, but there is no degree-$\ell$ $\sos$ certificate of infeasibility for $\ell\le \min\{2\lfloor k\rfloor+3, 2\lfloor n-k\rfloor+3,n\}$, as shown originally by Grigoriev and others recently using different approaches~\cite{grigoriev2001,potechin17}.\\

\noindent{\it The Sherali-Adams Hierarchy.}
There is a weaker hierarchy obtained using linear programming due to Sherali \& Adams (\SA)~\cite{SA90}.
Given disjoint subsets $S,R\subseteq E$, consider the polynomial $\varphi_{S,R}=\prod_{i\in S}x_i\prod_{j\in R}(1-x_j)$, and for every $\ell\in \{1,\ldots,|E|\}$, let $\mathcal{E}_{\ell}=\{(S,R):S,R\subseteq E \text{ with } |S\cup R| = \ell \text{ and } S\cap R=\emptyset\}$.
We say that a linear functional $\widetilde{\mathbb{E}}:\RR[x]/\I_E\to \RR$ is a degree-$\ell$ $\SA$ pseudoexpectation for (\ref{eq:semi1})-(\ref{eq:semi3}), if it satisfies the following properties:
\begin{enumerate}[label=\text{(SA.\arabic*)},align=right,leftmargin=1.8\labelwidth]
	\item \label{prop:SAnormal} $\pseudo(1)=1$,
	\item \label{prop:SA1} $\pseudo(\overline{\varphi_{S,R}})\ge 0$ for all $(S,R)\in \mathcal{E}_{\ell}$, 
	\item \label{prop:SA2} $\pseudo(\overline{\varphi_{S,R}g_i})\ge 0$ for all $i\in M$ and $(S,R)\in \mathcal{E}_{\ell}$ with $\deg(\overline{\varphi_{S,R}g_i})\le \ell$,
	\item \label{prop:SAA} $\pseudo(\overline{\varphi_{S,R}h_j})=0$ for all $j\in J$ and $(S,R)\in \mathcal{E}_{\ell}$ with $\deg(\overline{\varphi_{S,R}h_j})\le \ell$.
\end{enumerate}  
Observe that by construction it holds that every degree-$(\ell+1)$ $\SA$ pseudoexpectation is a degree-$\ell$ $\SA$ pseudoexpectation as well.
Furthermore, it follows directly from the linearity of $\pseudo$ that deciding whether a degree-$\ell$ $\SA$ pseudoexpectation exists,
and computing one if it exists,
can be done by solving a linear program of size $|E|^{O(\ell)}$ over the variables $y_S=\pseudo(x_{S})$, for every $S\subseteq E$ with $|S|\le \ell$.
This linear program is usually known as the {\it $\ell$-round or $\ell$-level} of the $\SA$ hierarchy.
For a detailed exposition of this hierarchy we refer to~\cite{Lau03}.
In the following we refer to {\it low-degree} when the degree ($\sos$ or $\SA$) of a certificate or pseudoexpectation is $O(1)$.

\section{Lower Bound: Symmetries are Hard for $\sos$}
\label{sec:lb}

In this section we show that the $\sos$ method fails to provide a low-degree certificate of infeasibility for a certain family of scheduling instances. 
The program we analize in this section is the configuration linear program, that has proven to be powerful for different scheduling and packing problems~\cite{ola12,GR15}.
Given a value $T>0$, a \emph{configuration}
corresponds to a multiset of processing times such that its total sum does not exceed $T$.
The multiplicity $m(p,C)$ indicates the number of times that the processing time $p$ appears in the multiset $C$.
The {\it load} of a configuration $C$ is just the total processing time, $\sum_{p\in \{p_j:j\in J\}}m(p,C)\cdot p$ and
let $\CC$ denote the set of all configurations with load at most $T$.
For each combination of a machine $i\in [m]$ and a configuration
$C\in \CC$, the program has a variable $y_{iC}$ that models whether
machine $i$ is scheduled with jobs with processing times according to configuration $C$.
Letting $n_p$ denote the number of jobs in $J$ with processing time $p$, we can write the following binary linear program, $\clp(T)$, 
\begin{alignat}{2}
	 \sum_{C\in \CC} y_{iC}                       & =   1   &&\quad\text{for all}\; i \in [m]\text{, } \\
  \sum_{i \in [m]} \sum_{C\in \CC} m(p,C)y_{iC}  & =   n_p &&\quad\text{for all}\; p \in \{  p_j: j\in J \} \text{, }\\
						     y_{iC}              & \in \{0,1\}   &&\quad\text{for all}\; i\in [m],\text{ for all } C\in \CC.
\end{alignat}

\noindent{\it Hard instances.} We briefly describe the construction of a family of hard instances $\{I_k\}_{k\in \NN}$ for the configuration linear program introduced in~\cite{KMMMVW18}.
Let $T=1023$, and for each odd $k\in \NN$ we have $15k$ jobs and $3k$ machines. 
There are 15 different job-sizes with value $O(1)$, each one with multiplicity $k$.
There exist a set of special configurations $\{C_1,\ldots,C_6\}$, called {\it matching configurations}, such that the program above is feasible if and only if the program restricted to the matching configurations is feasible.
The infeasibility of the latter program comes from the fact that there is no 1-factorization of a regular multigraph version of the Petersen graph~\cite[Lemma 2]{KMMMVW18}. 

\begin{theorem}[\hspace{-0.01cm}\cite{KMMMVW18}]
\label{thm:salb}
For each odd $k\in \NN$, there exists a degree-$\lfloor k/2\rfloor$ $\SA$ pseudoexpectation for the configuration linear program.
In particular, there is no low-degree $\SA$ certificate of infeasibility.
\end{theorem}



\subsection{A symmetry-reduced decomposition of the scheduling ideal}

Given $T>0$, the variables ground set for configuration linear program is $E=[m]\times \cal C$, and the symmetric group $S_m$ acts over the monomials in $\RR[y]$ according to $\sigma y_{iC}=y_{\sigma(i)C}$, for every $\sigma \in S_m$.
The action extends linearly to $\RR[y]/{\I}_E$, and the configuration linear program is invariant under this action, that is, for every $y\in \clp(T)$ and every $\sigma\in S_m$ we have $\sigma y\in \clp(T)$.
We say that a polynomial $f\in \RR[y]/\I_E$ is $S_m$-invariant if $\sigma f=f$ for every $\sigma\in S_m$. When it is clear from the context we drop the $S_m$ in the notation.
 If $f$ is invariant we have that $f=(1/|S_m|)\sum_{\sigma \in S_m}\sigma f:=\sym(f)$, which is the symmetrization of $f$.
We say that a linear functional $\cal L$ over the quotient ring is $S_m$-symmetric if for every polynomial $f\in \RR[y]/\I_E$ we have ${\cal L}(f)={\cal L}(\sym(f))$. The next lemma shows that when $\pseudo$ is symmetric it is enough to check symmetric polynomials in condition~\ref{prop:sos1}.
Therefore, in this case we restrict our attention to those polynomials that are invariant and $\sos$.

\begin{lemma}
\label{lem:sym-sos}
Let $\pseudo$ be a symmetric linear operator over $\RR[y]/\I_E$ such that for every invariant $\sos$ polynomial $g$ of degree at most $\ell$ we have $\pseudo(g)\ge 0$.
Then, $\pseudo(f^2)\ge 0$ for every $f\in \RR[y]/\I_E$ with $\deg(f^2)\le \ell$.
\end{lemma}

\begin{proof}
Since the operator $\pseudo$ is symmetric, for every $f$ in the quotient ring with $\deg(f^2)\le \ell$ we have $\pseudo(f^2)=\pseudo(\sym(f^2))$.
The polynomial $\sym(f^2)$ is symmetric, and it is $\sos$ since $\sym(f^2)=(1/|S_m|)\sum_{\sigma \in S_m}\sigma f^2$, which is a sum of squares.
Since $\deg(\sym(f^2))\le \ell$, we have $\pseudo(\sym(f^2))\ge 0$ and we conclude that $\pseudo(f^2)\ge 0$.
 \end{proof}

In the following we focus on understanding polynomials that are invariant and $\sos$. 
To analize the action of the symmetric group over $\RR[y]$ we introduce some tools from representation theory to characterize the invariant $S_m$-modules of the polynomial ring~\cite{sagan01}.
We maintain the exposition minimally enough for our purposes and we follow in part the notation used by Raymond et al.~\cite{RSST18}. 
We say that $V$ is an $S_m$-module if there exists a homomorphism $\rho:{S_m}\to \text{GL}(V)$, where $\text{GL}(V)$ is the linear group of $V$.
A subspace $W$ of $V$ is {\it invariant} if it is closed under the action of $S_m$, that is, when $w\in W$ and $\sigma \in S_m$ we have that $\sigma w\in W$. 
We say that an $S_m$-module $W$ is {\it irreducible} if the only invariant subspaces are $\{0\}$ and $W$.
We refer to~\cite{sagan01} for a deeper treatment of representation theory of symmetric groups.\\

\noindent{\it Isotypic decompositions.}
A {\it partition} of $m$ is a vector $(\lambda_1,\ldots,\lambda_t)$ such that $\lambda_1\ge \lambda_2\ge \cdots \lambda_t>0$ and $\lambda_1+\cdots+\lambda_t=m$.
We denote by $\lambda\vdash m$ when $\lambda$ is a partition of $m$.
Any $S_m$-module has an {\it isotypic decomposition} $V=\bigoplus_{\lambda\vdash m}V_{\lambda}$, which decomposes $V$ as a direct sum of $S_m$-modules,
where each of the subspaces in the direct sum is called an {\it isotypic component}.
In the following we introduce a combinatorial abstraction of the partitions and related subgroups that play a relevant role.
A {\it tableau} of shape $\lambda$ is a bijective filling between $[m]$ and the cells of a grid with $t$ rows, and every row $r\in [t]$ has length $\lambda_{r}$.
In this case, the {\it shape} or {\it Young diagram} of the tableau is $\lambda$.
For a tableau $\tau_{\lambda}$ of shape $\lambda$, we denote by $\row_{r}(\tau_{\lambda})$ the subset of $[m]$ that fills row $r$ in the tableau.
\begin{example}
Let $m=7$ and consider the partition $\lambda=(4,2,1)$.
The following tableaux have shape $\lambda$,
\vspace{.1cm}
\begin{center}
\begin{minipage}{0.3\linewidth}
\begin{tikzpicture}[scale=.4]
\node (v1) at (0,1) {};
\node (v2) at (1,0) {};
\node (v3) at (0,2) {};
\node (v4) at (2,1) {};
\node (v5) at (0,3) {};
\node (v6) at (4,2) {};

\node (v7) at (3,2) {};

\node (v11) at (0.5,0.5) {3};
\node (v22) at (0.5,1.5) {5};
\node (v33) at (0.5,2.5) {1};
\node (v44) at (1.5,1.5) {6};
\node (v55) at (1.5,2.5) {2};
\node (v66) at (2.5,2.5) {7};
\node (v77) at (3.5,2.5) {4};

\draw (v1) rectangle (v2) {};
\draw (v3) rectangle (v4) {};
\draw (v5) rectangle (v6) {};
\draw (v5) rectangle (v2) {};
\draw (v5) rectangle (v4) {};
\draw (v5) rectangle (v7) {};

\end{tikzpicture}
\end{minipage}
\begin{minipage}{0.3\linewidth}
\begin{tikzpicture}[scale=.4]
\node (v1) at (0,1) {};
\node (v2) at (1,0) {};
\node (v3) at (0,2) {};
\node (v4) at (2,1) {};
\node (v5) at (0,3) {};
\node (v6) at (4,2) {};

\node (v7) at (3,2) {};

\node (v11) at (0.5,0.5) {4};
\node (v22) at (0.5,1.5) {3};
\node (v33) at (0.5,2.5) {1};
\node (v44) at (1.5,1.5) {6};
\node (v55) at (1.5,2.5) {7};
\node (v66) at (2.5,2.5) {2};
\node (v77) at (3.5,2.5) {5};

\draw (v1) rectangle (v2) {};
\draw (v3) rectangle (v4) {};
\draw (v5) rectangle (v6) {};
\draw (v5) rectangle (v2) {};
\draw (v5) rectangle (v4) {};
\draw (v5) rectangle (v7) {};

\end{tikzpicture}
\end{minipage}
\end{center}
In the tableau $\tau_{\lambda}$ at the left, $\row_1(\tau_{\lambda})=\{1,2,7,4\}$.
In the tableau $\tau'_{\lambda}$ at the right, $\row_3(\tau'_{\lambda})=\{4\}$.
\end{example}

\noindent The row group $\rowgroup$ is the subgroup of $S_m$ that stabilizes the rows of the tableau $\tau_{\lambda}$, that is,
\begin{equation}
\label{eq:row-group}
\rowgroup=\Big\{\sigma\in S_m:\sigma \cdot \row_{r}(\tau_{\lambda})=\row_{r}(\tau_{\lambda})\text{ for every $r\in [t]$}\Big\}.
\end{equation}
\noindent{\it Invariant $\sos$ polynomials.} We go back now to the case of the configuration linear program.
\begin{definition}[Scheduling Ideal]
\label{def:sched-ideal}
We define $\sched$ to be the ideal of polynomials in $\RR[y]$ generated by 
\begin{equation}
\left\{\sum_{C\in \CC}y_{iC}-1:i\in [m]\right\}\cup \Big\{y_{iC}^2-y_{iC}:i\in [m], C\in \CC\Big\}.
\end{equation}
\end{definition}
Recall that the set of polynomials above enforce the machines in the scheduling solutions to be assigned with exactly one configuration.
Let $\bf{Q}^{\ell}$ be the quotient ring $\RR[y]/\sched$ restricted to polynomials of degree at most $\ell$ and let $\bigoplus_{\lambda\vdash m}{\bf{Q}^{\ell}_{\lambda}}$ be its isotypic decomposition.
Given a tableau $\tau_{\lambda}$ of shape $\lambda$, let $\rowspace^{\ell}$ be the subspace of ${\bf{Q}}^{\ell}_{\lambda}$ fixed by the action of the row group $\rowgroup$, that is,
\begin{equation}
\label{eq:iso-row}
\rowspace^{\ell}=\Big\{q\in {\bf{Q}^{\ell}_{\lambda}}:\sigma q=q \text{ for all }\sigma\in \rowgroup\Big\}.
\end{equation}
In what follows we sometimes refer to these subspaces as {\it row subspaces}.
The following result follows from the work of Gaterman \& Parrilo~\cite{GP04} in the context of symmetry reduction for invariant semidefinite programs.
In what follows, $\langle A,B\rangle$ is the inner product in the space of square matrices defined by the trace of $AB$.
Given $\ell\in [m]$, we denote by $\Lambda_{\ell}$ the subset of partitions of $m$ that are lexicographically larger than $(m-\ell,1,\ldots,1)$.
\begin{theorem}
\label{thm:GP}
Suppose that $g\in \RR[y]/\sched$ is a degree-$\ell$ $\sos$ and $S_m$-invariant polynomial.
For each partition $\lambda\in \Lambda_{\ell}$, let $\tau_{\lambda}$ be a tableau of shape $\lambda$ and let ${\cal P}^{\lambda}=\{p_1^{\tau_\lambda},\ldots,p_{n_{\lambda}}^{\tau_\lambda}\}$ be a set of polynomials such that $\text{span}({\cal P}^{\lambda})\supseteq \rowspace^{\ell}$.
Then, for each partition $\lambda\in \Lambda_{\ell}$ there exists a 
positive semidefinite matrix $M_{\lambda}$ such that 
$g=\sum_{\lambda \in \Lambda_{\ell}}\langle M_{\lambda},Z^{\tau_\lambda}\rangle$, where $Z^{\tau_\lambda}_{ij}=\sym(p_i^{\tau_\lambda}p_j^{\tau_\lambda})$.
\end{theorem}
The theorem above is based on the recent work of Raymond et al.~\cite[p. 324, Theorem 3]{RSST18}. 
In our case the symmetric group is acting differently from Raymond et al., but the proof follows the same lines,
and it can be found in Appendix~\ref{sec:appendixA}.
Together with Lemma~\ref{lem:sym-sos}, it is enough to study pseudoexpectations for each of the partitions in $\Lambda_{\ell}$ separately. 
We remark that for each partition in $\lambda\in \Lambda_{\ell}$ we can take any tableau $\tau_{\lambda}$ with that shape, and then consider a spanning set for its corresponding subspace $\rowspace^{\ell}$.
In the following, for a matrix $A$ with entries in $\RR[y]$, we denote by $\pseudo(A)$ the matrix obtained by evaluating $\pseudo$ on each entry of $A$.

\begin{lemma}
\label{lem:row-pseudo}
Let $\pseudo$ be a symmetric linear operator over $\RR[y]/\I_E$.
For each $\lambda\in \Lambda_{\ell}$, let $\tau_{\lambda}$ be a tableau of shape $\lambda$ and let ${\cal P}^{\lambda}=\{p_1^{\tau_\lambda},\ldots,p_{n_{\lambda}}^{\tau_\lambda}\}$ be a set of polynomials such that $\text{span}({\cal P}^{\lambda})\supseteq \rowspace^{\ell}$.
For each $\lambda \in \Lambda_{\ell}$, let $Z^{\tau_\lambda}$ such that $Z^{\tau_\lambda}_{ij}=\sym(p_i^{\tau_\lambda}p_j^{\tau_\lambda})$ and
suppose that $\pseudo(Z^{\tau_\lambda})$ is positive semidefinite.
Then, $\pseudo(f^2)\ge 0$ for every $f\in \RR[y]/\I_E$ with $\deg(f^2)\le \ell$.
\end{lemma}

\begin{proof}
Let $g$ be an invariant $\sos$ polynomial of degree at most $\ell$. 
By Theorem~\ref{thm:GP}, for each $\lambda\in \Lambda_{\ell}$ there exist a positive semidefinite matrix $M_{\lambda}$ such that $g=\sum_{\lambda \in \Lambda_{\ell}}\langle M_{\lambda},Z^{\lambda}\rangle$. 
Therefore, we have that 
\[\pseudo(g)=\sum_{\lambda \in \Lambda_{\ell}}\pseudo\langle M_{\lambda},Z^{\lambda}\rangle=\sum_{\lambda \in \Lambda_{\ell}}\langle M_{\lambda},\pseudo(Z^{\lambda})\rangle\ge 0,\] since both $M_{\lambda}$ and $\pseudo(Z^{\lambda})$ are positive semidefinite for each partition $\lambda\in \Lambda_{\ell}$.  
By Lemma~\ref{lem:sym-sos} we conclude that $\pseudo(f^2)\ge 0$ for every $f\in \RR[y]/\I_E$ with $\deg(f^2)\le \ell$.
 \end{proof}

\subsection{Construction of the spanning sets}

In this section we show how to construct the spanning sets of the row subspaces in order to apply Lemma~\ref{lem:row-pseudo}, which together with a particular linear operator provides the existence of a high-degree $\sos$ pseudoexpectation. 
The structure of the configuration linear program allows us to further restrict the canonical spanning set obtained from monomials, by one that is combinatorially interpretable and adapted to our purposes.


\begin{definition}[Partial Schedule]
\label{def:partial-schedule}
Let $G_S$ be the directed bipartite graph with vertex partition given by $[m]$ and $\CC$ and edges~$S\subseteq [m]\times \CC$.
We say that $S\subseteq [m]\times \CC$ is a {\it partial schedule} if for every $i\in [m]$ we have $\delta_S(i)\le 1$, where $\delta_S(i)$ is the degree of vertex $i$ in $G_S$.
\end{definition}

\noindent We say that $S$ is a partial schedule over $H$ if $\{i\in [m]:(i,C)\in S\}\subseteq H$. 
We denote by $\mathcal{M}(S)$ the set of machines 
in $\{i\in [m]:\delta_S(i)=1\}$, and we call $\mathcal{M}(S)$ the set of machines {\it incident} to $S$. 
Sometimes it is convenient to see a partial schedule $S$ as a function from $\mathcal{M}(S)$ to $\CC$, so we also say that $S$ is partial schedule with domain $\mathcal{M}(S)$.
\begin{example}
Let $m=4$ and the set of configurations $\CC=\{C_1,C_2,C_3\}$.
Then, the set given by $T=\{(1,C_1),(2,C_1),(4,C_2)\}$ is a partial schedule. The machine $i=3$ is not incident to $T$. 
In this case, $\delta_T(C_1)=2$ since there are two machines, $\{1,2\}$, incident to $C_1$. 
The domain of $T$ is $\mathcal{M}(T)=\{1,2,4\}$.
The set $S=\{(1,C_1),(1,C_2)\}$ is not a partial schedule since $\delta_S(1)=2$.
\end{example}


\begin{proposition}
\label{lem:not-partial-vanish}
If $S\subseteq [m]\times \CC$ is not a partial schedule, we have $y_S\equiv 0\mod \sched$.
\end{proposition}

\begin{proof}
Since $S$ it is not a partial schedule, there exists a machine $i\in [m]$ such that $\delta_S(i)\ge 2$. 
Therefore, to prove the proposition it is enough to check that $y_{iC}y_{iR}\equiv 0\mod \sched$ for every pair of different configurations $R, C\in \CC$.
Given a configuration $C\in \CC$, we have that
\[\sum_{R\in \CC\setminus \{C\}}y_{iC}y_{iR}\equiv \sum_{S\in \CC\setminus \{C\}}y_{iC}y_{iR}+y_{iC}^2-y_{iC}
\equiv y_{iC}(\sum_{R\in \CC}y_{iS}-1)\equiv 0\mod \sched.\]
On the other hand, $y_{iC}^2y_{iR}^2\equiv y_{iC}y_{iR}$ for every $R\in \CC\setminus \{C\}$.
This yields the result.
 \end{proof}

\begin{proposition}
\label{lem:deg-blowup}
Let $S\subseteq [m]\times \CC$ be a partial schedule of cardinality at most $\ell$. 
Then, 
$y_S\in \text{span}(\{y_L:|L|= \ell\text{ and S is a partial schedule}\}).$
\end{proposition}

\begin{proof}
Assume that $|S|<\ell$ since otherwise we are done.
Let $H\subseteq [m]$ such that $|H|=\ell-|S|$ and $\delta_S(h)=0$ for every $h\in H$, that is, $H$ is subset of machines that is not incident to the edges $S$ in the bipartite graph $G_S$.
Observe that since $S$ is a partial schedule, it is incident to exactly $|S|$ machines.
Let $\CC^{H}$ be the set of partial schedules with domain $H$.
Since $\sum_{C\in \CC}y_{hC}\equiv 1\mod \sched$ for every $h\in H$, we have 
\[y_S\equiv y_S\prod_{h\in H}\sum_{C\in \CC}y_{hC} \equiv \sum_{R\in \CC^{H}} y_{S\cup R}\mod \sched.\]
In particular, for every $R\in \CC^{H}$ we have that $S\cup R$ is a partial schedule, and $\deg(y_{S\cup R})=|S|+\ell-|S|=\ell$.
 \end{proof}

In the following we construct spanning sets for the row subspaces. 
Given a tableau $\tau_{\lambda}$ with shape $\lambda$, the $\hook(\tau_{\lambda})$ is the tableau with shape $(\lambda_1,1,\ldots,1)\in \ZZ^{m-\lambda_1+1}$, its first row it is equal to the first row of $\tau_{\lambda}$ and the remaining elements of $\tau_{\lambda}$ fill the rest of the cells in increasing order over the rows. 
That part is called the {\it tail} of the hook, and we denote by $\tail(\tau_{\lambda})$ the elements of $[m]$ in the tail of $\hook(\tau_{\lambda})$, and $\row(\tau_{\lambda})=[m]\setminus \tail(\tau_{\lambda})$, that is the elements in the first row of the tableau.

\begin{example}
Let $m=7$ and consider the partition $\lambda=(4,2,1)$.
The tableau $\tau_{\lambda}$ at the left has shape $\lambda$ and the tableau at the right is $\hook(\tau_{\lambda})$, with shape $(4,1,1,1)$; $\row(\tau_{\lambda})=\{1,2,7,4\}$ and $\tail(\tau_{\lambda})=\{3,5,6\}$.
\vspace{.1cm}
\begin{center}
\begin{minipage}{0.3\linewidth}
\begin{tikzpicture}[scale=.4]
\node (v1) at (0,1) {};
\node (v2) at (1,0) {};
\node (v3) at (0,2) {};
\node (v4) at (2,1) {};
\node (v5) at (0,3) {};
\node (v6) at (4,2) {};

\node (v7) at (3,2) {};

\node (v11) at (0.5,0.5) {3};
\node (v22) at (0.5,1.5) {5};
\node (v33) at (0.5,2.5) {1};
\node (v44) at (1.5,1.5) {6};
\node (v55) at (1.5,2.5) {2};
\node (v66) at (2.5,2.5) {7};
\node (v77) at (3.5,2.5) {4};

\draw (v1) rectangle (v2) {};
\draw (v3) rectangle (v4) {};
\draw (v5) rectangle (v6) {};
\draw (v5) rectangle (v2) {};
\draw (v5) rectangle (v4) {};
\draw (v5) rectangle (v7) {};

\end{tikzpicture}
\end{minipage}
\begin{minipage}{0.3\linewidth}
\begin{tikzpicture}[scale=.4]
\node (v1) at (1,-1) {};
\node (v2) at (1,0) {};
\node (v3) at (1,1) {};
\node (v4) at (2,2) {};
\node (v5) at (3,2) {};
\node (v6) at (4,2) {};

\node (v7) at (0,3) {};

\node (v11) at (0.5,-0.5) {6};
\node (v22) at (0.5,0.5) {5};
\node (v33) at (0.5,1.5) {3};
\node (v44) at (0.5,2.5) {1};
\node (v55) at (1.5,2.5) {2};
\node (v66) at (2.5,2.5) {7};
\node (v77) at (3.5,2.5) {4};

\draw (v1) rectangle (v7) {};
\draw (v2) rectangle (v7) {};
\draw (v3) rectangle (v7) {};
\draw (v4) rectangle (v7) {};
\draw (v5) rectangle (v7) {};
\draw (v6) rectangle (v7) {};

\end{tikzpicture}
\end{minipage}
\end{center}
\end{example}

The following lemma gives a spanning set for the row subspaces obtained from the hook tableau. 
We denote by $\sym_{\hook(\tau_{\lambda})}$ the symmetrization respect to the row subgroup of $\hook(\tau_{\lambda})$,
\begin{equation}
\sym_{\hook(\tau_{\lambda})}(f)=\frac{1}{|{\mathcal{R}}_{\hook(\tau_{\lambda})}|}\sum_{\sigma \in {\bf{R}}_{\hook(\tau_{\lambda})}}\sigma f.
\end{equation}
The following lemma provides a spanning set for the row subspace based on the above family polynomials.
The proof follows the lines of~\cite[Lemma 2]{RSST18}.
\begin{lemma}
\label{lem:hooks}
Given a tableau $\tau_{\lambda}$, the row subspace $\rowspace^{\ell}$ of ${\bf{Q}}^{\ell}$ is spanned by
\begin{equation}
\label{eq:hooks}
\Big\{\sym_{\hook(\tau_{\lambda})}(y_S):|S|= \ell\text{ and S is a partial schedule}\Big\}.
\end{equation}
\end{lemma}
\begin{proof}
Let $\mathcal{A}=\{q\in {\bf{Q}^{\ell}}:\sigma q=q \text{ for all }\sigma\in \rowgroup\}$ and $\mathcal{A}'=\{q\in {\bf{Q}^{\ell}}:\sigma q=q \text{ for all }\sigma\in \calR_{\hook(\tau_{\lambda})}\}$.
By definition, we have that $\rowspace^{\ell}\subseteq \mathcal{A}$, and since $\calR_{\hook(\tau_{\lambda})}$ is a subgroup of $\rowgroup$ it follows that $\mathcal{A}\subseteq \mathcal{A}'.$
By Propositions~\ref{lem:not-partial-vanish} and~\ref{lem:deg-blowup}, and the linearity of the symmetrization operator, we have that $\mathcal{A}'$ is spanned by the set in~(\ref{eq:hooks}). 
 \end{proof}

In the row subgroup ${\mathcal{R}}_{\hook(\tau_{\lambda})}$, the elements of $[m]$ that are in the tail remain fixed. 
The rest of the elements on the first row are permuted arbitrarily. 
In particular, ${\mathcal{R}}_{\hook(\tau_{\lambda})}\cong S_{\lambda_1}$.
Therefore, any permutation $\sigma$ in ${\mathcal{R}}_{\hook(\tau_{\lambda})}$ acts over a monomial $y_S$ by separating the bipartite graph $G_S$ into those vertices in $\tail(\tau_{\lambda})$ that are fixed by $\sigma$ and the rest in $\row(\tau_{\lambda})$ that can be permuted. \\

\noindent{\it Configuration profiles.}
Observe that bipartite graphs corresponding to different partial schedules are isomorphic if and only if the degree of every configuration is the same in both graphs. 
We say that a partial schedule is in {\it $\gamma$-profile}, with $\gamma:\CC\to \mathbb{Z}_+$, if for every $C\in \CC$ we have $\delta_S(C)=\gamma(C)$. 
Observe that a partial schedule in $\gamma$-profile has size $\sum_{C\in \CC}\gamma(C)$, quantity that we denote by $\|\gamma\|$.
We denote by $\text{supp}(\gamma)$ the support of the vector $\gamma$, namely, $\{C\in \CC:\gamma(C)>0\}$.
\begin{definition}
Given a partial schedule $T$, we say that a partial schedule $A$ over $[m]\setminus \mathcal{M}(T)$
is a $(T,\gamma)$-extension if 
$A$ is in $\gamma$-profile.
We denote by ${\cal F}(T,\gamma)$ the set of $(T,\gamma)$-extensions.
In particular, every $(T,\gamma)$-extension has size $\|\gamma\|$.
\end{definition}

\begin{example}
Consider $m=4$, $\CC=\{C_1,C_2\}$ and the partial schedule $T=\{(2,C_1),(3,C_2)\}$.
If $\gamma$ is given by $\gamma(C_1)=\gamma(C_2)=1$, we have $\mathcal{F}(T,\gamma)=\{\{(1,C_1),(4,C_2)\},\{(4,C_1),(1,C_2)\}\}$.
If $\mu$ is given by $\mu(C_1)=1$ and $\mu(C_2)=0$, we have $\mathcal{F}(T,\mu)=\{\{(1,C_1)\},\{(4,C_1)\}\}$.
\end{example}

\noindent Given a partial schedule $T$ and a $\gamma$-profile, let $\mathcal{B}_{T,\gamma}$ be the polynomial defined by
\begin{equation}
\label{def:key-poly}
\mathcal{B}_{T,\gamma}=\sum_{A\in {\cal F}(T,\gamma)}y_A,
\end{equation}
if $\gamma\ne 0$, and $1$ otherwise. 
In words, the polynomial above corresponds to sum over all those partial schedules in $\gamma$-profile that are not incident to $\mathcal{M}(T)$.
The following theorem is the main result of this section.
\begin{theorem}
\label{thm:tails}
Let $\lambda\in \Lambda_{\ell}$ and a tableau $\tau_{\lambda}$ of shape $\lambda$.
Then, the row subspace $\rowspace^{\ell}$ of ${\bf{Q}}^{\ell}$ is spanned by
\begin{equation}
\label{eq:hooks-sched}
{\cal{P}}^{\lambda}=\bigcup_{\omega=0}^{\ell}\Big\{y_T\mathcal{B}_{T,\gamma}:T\text{ is partial schedule with $\mathcal{M}(T)=\text{tail}(\tau_{\lambda})$ and }\|\gamma\|=\omega\Big\}.
\end{equation}
\end{theorem}



\begin{proof}
By Lemma~\ref{lem:hooks} it is enough to check that the set of polynomials in (\ref{eq:hooks}) is spanned by those in (\ref{eq:hooks-sched}).
Let $S$ be a partial schedule of size $\ell$.
Let $\tail(S,\tau_{\lambda})$ be the subset of $S$ that is incident to the tail of the tableau, that is, $\{(i,C)\in S:i\in \tail(\tau_{\lambda})\}$, and let $\row(S,\tau_{\lambda})=S\setminus \tail(S,\tau_{\lambda})$ be the edges of the partial schedule $S$ incident to the first row of the tableau.

\begin{claim}
\label{claim:tail-and-hooks}
$\sym_{\hook(\tau_{\lambda})}(y_S)=y_{\tail(S,\tau_{\lambda})}\cdot \sym_{\hook(\tau_{\lambda})}\Big(y_{\row(S,\tau_{\lambda})}\Big).$
\end{claim}

Observe that $\tail(S,\tau_{\lambda})$ is a partial schedule over $\tail(\tau_{\lambda})$. 
Similarly as we did in Lemma~\ref{lem:deg-blowup}, the partial schedule incident to the tail can be completed to be in the span of partial schedules with domain equal to $\tail(\tau_{\lambda})$, that is,
\begin{align*}
y_{\tail(S,\tau_{\lambda})}&\equiv y_{\tail(S,\tau_{\lambda})}\prod_{h\in \tail(\tau_{\lambda})\setminus \tail(S,\tau_{\lambda})}\sum_{C\in \CC}y_{hC}\equiv \sum_{L\in \CC^{\tail(\tau_{\lambda})\setminus \tail(S,\tau_{\lambda})}} y_{\tail(S,\tau_{\lambda})\cup L}\mod \sched
\end{align*}
where $\CC^{\tail(\tau_{\lambda})\setminus \tail(S,\tau_{\lambda})}$ is the set of partial schedules with domain $\tail(\tau_{\lambda})\setminus \tail(S,\tau_{\lambda})$.
Thus, every partial schedule in the summation above have domain $\tail(\tau_{\lambda})\cup \tail(S,\tau_{\lambda})\setminus \tail(S,\tau_{\lambda})=\tail(\tau_{\lambda})$.
Therefore, it is enough to check that exists a constant $\kappa$ such that $\sym_{\row(\tau_{\lambda})}(y_{\row(S,\tau_{\lambda})})=\kappa\cdot \mathcal{B}_{\tail(\tau_{\lambda}),\gamma}$
for some profile $\gamma$ with $\|\gamma\|=\ell-|\tail(S,\tau_{\lambda})|$.
Recall that $|\tail(S,\tau_{\lambda})|\le \ell$ since $\lambda\in \Lambda_{\ell}$.
Let $\gamma$ be the profile of the partial schedule $\row(S,\tau_{\lambda})$.
The equality follows since $\sigma\in{\mathcal{R}}_{\text{hook}(\tau_{\lambda})}\cong S_{\row(\tau_{\lambda})}$, together with the fact that $\{(\sigma(i),C):(i,C)\in \row(S,\tau_{\lambda})\}$ is a $(\tail(\tau_{\lambda}),\gamma)$-extension for every permutation in $\sigma\in{\mathcal{R}}_{\text{hook}(\tau_{\lambda})}$.
The constant $\kappa$ is equal to $|{\mathcal{R}}_{\text{hook}(\tau_{\lambda})}|$.
 \end{proof}

\begin{proof}[Claim~\ref{claim:tail-and-hooks}]
Observe that for every permutation $\sigma\in{\mathcal{R}}_{\hook(\tau_{\lambda})}$, we have 
\begin{align*}
\sigma y_S&
=\prod_{(i,C)\in\tail(S,\tau_{\lambda})}y_{\sigma(i)C}\prod_{(i,C)\in\row(S,\tau_{\lambda})}y_{\sigma(i)C}
=y_{\tail(S,\tau_{\lambda})}\sigma y_{\row(S,\tau_{\lambda})},
\end{align*}
since the permutation fixes the edges in $\tail(S,\tau_{\lambda})$.
Therefore, symmetrizing yields that $\sym_{\hook(\tau_{\lambda})}(y_S)$ is equal to
\begin{align*}
\frac{1}{|{\mathcal{R}}_{\hook(\tau_{\lambda})}|}\sum_{\sigma \in {\mathcal{R}}_{\hook(\tau_{\lambda})}}\sigma y_S
						&=y_{\tail(S,\tau_{\lambda})}\cdot \frac{1}{|{\mathcal{R}}_{\hook(\tau_{\lambda})}|}\sum_{\sigma \in {\mathcal{R}}_{\hook(\tau_{\lambda})}}\sigma y_{\row(S,\tau_{\lambda})}\\
						&=y_{\tail(S,\tau_{\lambda})}\cdot \sym_{\hook(\tau_{\lambda})}\Big(y_{\row(S,\tau_{\lambda})}\Big).	\qedhere
\end{align*}
\end{proof}



\subsection{High-degree SoS pseudoexpectation: Proof of Theorem~\ref{thm:negative}}

We now have the ingredients to study the scheduling ideal and we describe the pseudoexpectations from Theorem~\ref{thm:salb}, that are the base for our lower bound.
Recall that for every odd $k\in \NN$, the hard instance $I_k$ has $m=3k$ machines and the linear operators we consider are supported over partial schedules incident to a set of six so called {\it matching configurations}, $\{C_1,\ldots,C_6\}$.
Consider the $\pseudo:\RR[y]/\I_E\to \RR$ such that for every partial schedule $S$ of cardinality at most $k/2$, 
\begin{equation}
\label{eq:pseudo-old}
\pseudo(y_S)=\frac{1}{(3k)_{|S|}}\prod_{j=1}^{6}(k/2)_{\delta_S(C_j)},
\end{equation}
where $(a)_b$ is the lower factorial function, that is, $(a)_b=a(a-1)\cdots (a-b+1)$, and $(a)_0=1$. 
The linear operator $\pseudo$ is zero elsewhere.
We state formally the main result that implies Theorem~\ref{thm:negative}.

\begin{theorem}
\label{thm:main-negative}
For every odd $k\in \NN$, the linear operator $\pseudo$ is a degree-$\lfloor k/6\rfloor$ $\sos$ pseudoexpectation for the configuration linear program in instance $I_k$ and $T=1023$.
\end{theorem}

\begin{proof}[Theorem~\ref{thm:negative}]
For every odd $k$ the instance $I_k$ described in Section~\ref{sec:lb} is infeasible for $T=1023$.
By Theorem~\ref{thm:main-negative}, the operator $\pseudo$ is a degree-$\lfloor k/6\rfloor$ $\sos$ pseudoexpectation, which in turns imply by Lemma~\ref{lem:infeasible-pseudo} that there is no degree-$\lfloor k/6\rfloor$ $\sos$ certificate of infeasibility.
For an instance with $n$ jobs, let $k$ be the greatest odd integer such that $n=15k+\ell$, with $\ell<30$.
The theorem follows by considering the instance $I_k$ above with $\ell$ dummy jobs of processing time equal to zero.
 \end{proof}

Theorem~\ref{thm:salb} guarantees that for every $k$ odd, $\pseudo$ is a degree-$\lfloor k/2\rfloor$ pseudoexpectation, and therefore a degree-$\lfloor k/6\rfloor$ pseudoexpectation as well.
In particular, properties \ref{prop:normal} and \ref{prop:SA} are satisfied. 
Since the configuration linear program is constructed from equality constraints, it is enough to check property \ref{prop:sos1} for high enough degree, in this case $\ell=\lfloor k/6\rfloor$.
To check property \ref{prop:sos1} we require a notion of {\it conditional pseudoexpectations}.
Given a partial schedule $T$, consider the operator $\pseudo_{T}:\RR[y]/\I_E\to \mathbb{R}$ such that
\begin{equation}
\label{def:conditional-pseudo}
\pseudo_{T}(y_S)=\frac{1}{(3k-|T|)!}\prod_{j=1}^{6}(k/2-\delta_T(C_j))_{\delta_S(C_j)}
\end{equation}
for every partial schedule $S$ over the machines $[m]\setminus \mathcal{M}(T)$ and zero otherwise.
Observe that if $T=\emptyset$ it corresponds to the linear operator $\pseudo$ in (\ref{eq:pseudo-old}).
The following lemmas about the conditional pseudoexpectation in (\ref{def:conditional-pseudo}) are key for proving that $\pseudo$ is a high-degree $\sos$ pseudoexpectation. 
We state the lemmas and show how to conclude Theorem~\ref{thm:main-negative} using them.
In particular, in Lemma~\ref{lem:key-lem} we prove a strong {\it pseudoindependence} property satisfied by the conditional pseudoexpectations and the polynomials (\ref{def:key-poly}) in the spanning set. 

\begin{lemma}
\label{lem:pseudo-sym}
The linear operator $\pseudo$ is $S_m$-symmetric.
\end{lemma}

\begin{lemma}
\label{lem:conditioning}
Let $T$ be a partial schedule. 
Then, the following holds:
\begin{enumerate}
	\item[$(a)$]\label{lem:conditioning-a} If $S$ is a partial schedule and $T\cap S=\emptyset$, then $\pseudo(y_Ty_S)=\pseudo_T(y_S)\pseudo(y_T)$.
	\item[$(b)$]\label{lem:conditioning-b} If $S,R$ are two partial schedules such that $R\cap S=\emptyset$ and $T\cap (R\cup S)=\emptyset$, then 
	\[\pseudo_T(y_Ry_S)=\pseudo_T(y_R)\cdot \pseudo_{T\cup R}(y_S).\]
	\item[$(c)$]\label{lem:conditioning-d} Let $\gamma$ be a profile with $\text{supp}(\gamma)\subseteq \{C_1,\ldots,C_6\}$ and $|T|+\|\gamma\|\le k/2$. 
	Then,
	\[\pseudo_{T}({\mathcal{B}_{T,\gamma}})=\prod_{j=1}^{6}\frac{1}{\gamma(C_j)!}(k/2-\delta_T(C_j))_{\gamma(C_j)}.\] 
\end{enumerate}
\end{lemma}

\begin{lemma}
\label{lem:key-lem}
Let $T$ be a partial schedule and  $\gamma$, $\mu$ a pair of configuration profiles with $|T|+\|\gamma\| +\|\mu\| \le k/2$ and $\text{supp}(\gamma),\text{supp}(\mu)\subseteq \{C_1,\ldots,C_6\}$.
Then,
\begin{equation}
\label{eq:truco}
\pseudo_T(\mathcal{B}_{T,\gamma}\mathcal{B}_{T,\mu})=\pseudo_T(\mathcal{B}_{T,\gamma})\cdot \pseudo_T(\mathcal{B}_{T,\mu}).
\end{equation}
\end{lemma}


\begin{proof}[Theorem~\ref{thm:main-negative}]
Let $\ell=\lfloor k/6\rfloor$.
Given a partition $\lambda\in \lambda_{\ell}$, consider the tableau $\tau_{\lambda}$ such that $\tail(\tau_{\lambda})=[3k-\lambda_1]$ and $\row(\tau_{\lambda})=[3k]\setminus [3k-\lambda_1]$.
The partial schedules with domain $[3k-\lambda_1]$ can be identified with $\CC^{[3k-\lambda_1]}$, the set of functions from $[3k-\lambda_1]$ to $\CC$.
In particular the spanning set in (\ref{eq:hooks-sched}) is described by
$\mathcal{P}^{\lambda}=\bigcup_{\omega=0}^{\ell}\Big\{y_T\beta_{T,\gamma}:T\in \CC^{[3k-\lambda_1]}\text{ and }\|\gamma\|=\omega\Big\}.$
To apply Lemma~\ref{lem:row-pseudo} we need to study the matrix $\pseudo(Z^{\lambda})$. 
Recall that for $T,S\in \CC^{[3k-\lambda_1]}$ and profiles $\gamma,\nu$ with $\|\gamma\|,\|\mu\|\le \ell$, the corresponding entry of the matrix $\pseudo(Z^{\lambda})$ is given by
\[
\pseudo\left(\sym\Big(y_Ty_{S}\beta_{T,\gamma}\beta_{S,\mu}\Big)\right)=\pseudo\left(\sym\Big(y_{T\cup S}\beta_{T,\gamma}\beta_{S,\mu}\Big)\right).\]
By Lemma \ref{lem:pseudo-sym} the operator $\pseudo$ is symmetric, and therefore,
\[\pseudo\left(\sym\Big(y_{T\cup S}\beta_{T,\gamma}\beta_{S,\mu}\Big)\right)=\pseudo\Big(y_{T\cup S}\beta_{T,\gamma}\beta_{S,\mu}\Big).\]
Since both $T,S$ are partial schedules such that $\mathcal{M}(T)=\mathcal{M}(S)$, we have that $T\cup S$ is a partial schedule if and only if $T=S$.
Thus, the matrix $\pseudo(Z^{\lambda})$ is block diagonal, with a block for each partial schedule $T\in \CC^{[3k-\lambda_1]}$.
For every $\Theta$ indexed by the elements of the spanning set above, we have then
\[
\Big\langle\pseudo(Z^{\lambda}),\Theta\Theta^{\top}\Big\rangle=\sum_{T\in \CC^{[3k-\lambda_1]}}\sum_{\substack{\gamma:\|\gamma\|\le \ell\\ \mu:\|\mu\|\le \ell}}\pseudo\Big(y_T\beta_{T,\gamma}\beta_{T,\mu}\Big)\Theta_{T,\gamma}\Theta_{T,\mu}.\]
Since $|T|+\|\gamma\|+\|\mu\|\le 3\ell\le k/2$ for every partial schedule $T$ and profiles $\gamma,\mu$ as above, by applying Lemma~\ref{lem:conditioning} $(a)$ and Lemma~\ref{lem:key-lem} we obtain that 
\begin{align*}
&\sum_{T\in \CC^{[3k-\lambda_1]}}\sum_{\substack{\gamma:\|\gamma\|\le \ell\\ \mu:\|\mu\|\le \ell}}\pseudo\Big(y_T\beta_{T,\gamma}\beta_{T,\mu}\Big)\Theta_{T,\gamma}\Theta_{T,\mu}
\\
&=\sum_{T\in \CC^{[3k-\lambda_1]}}\pseudo(y_T)\sum_{\substack{\gamma:\|\gamma\|\le \ell\\ \mu:\|\mu\|\le \ell}}\pseudo_T(\beta_{T,\gamma})\pseudo_T(\beta_{T,\mu})\Theta_{T,\gamma}\Theta_{T,\mu},
\end{align*}
and by rearranging terms we conclude that
\[
\Big\langle\pseudo(Z^{\lambda}),\Theta\Theta^{\top}\Big\rangle=\sum_{T\in \CC^{[3k-\lambda_1]}}\pseudo(y_T)\left(\sum_{\gamma:\|\gamma\|\le \ell}\pseudo_T(\beta_{T,\gamma})\Theta_{T,\gamma}\right)^2\ge 0.\]\qedhere
\end{proof}

\begin{proof}[Lemma~\ref{lem:pseudo-sym}]
Given $\sigma\in S_m$ and a partial schedule $S$, $\pseudo(\sigma y_S)=\pseudo(y_{\sigma(S)})$, where $\sigma(S)=\{(\sigma(i),C):(i,C)\in S\}$.
In particular, since $|S|=|\sigma(S)|$ and profile of $S$ is the same profile of $\sigma(S)$, it holds $\pseudo(y_S)=\pseudo(\sigma y_S)$.
Therefore, $\pseudo(y_S)=\frac{1}{m!}\sum_{\sigma\in S_m}\pseudo(\sigma y_S)=\pseudo(\sym(y_S))$.
 \end{proof}

\begin{proof}[Lemma~\ref{lem:conditioning}]
Property $b)$ implies $a)$ by taking $T=\emptyset$.
One can check from the definition of the lower factorial that $(x)_{a+b}=(x)_a(x-a)_b$.
Since the partial schedules $R,S$ and $T$ are disjoint, it holds for every $C\in \CC$ that $\delta_{R\cup S}(C)=\delta_{R}(C)+\delta_{S}(C)$ and $\delta_{T\cup R}(C)=\delta_{T}(C)+\delta_{R}(C)$.
Therefore,
\begin{align*}
&(3k-|T|)_{|R\cup S|}\cdot \pseudo_T(y_Ry_S)\\
				&=\prod_{j=1}^{6}(k/2-\delta_T(C_j))_{\delta_{R}(C_j)}\cdot \prod_{j=1}^{6}(k/2-\delta_T(C_j)-\delta_{R}(C_j))_{\delta_{S}(C_j)}\\
				&=(3k-|T|)_{|R|}\cdot \pseudo_T(y_R)\cdot (3k-|T|-|R|)_{|S|}\cdot \pseudo_{T\cup R}(y_S),
\end{align*}
and the lemma follows since $(3k-|T|)_{|R\cup S|}=(3k-|T|)_{|R|}\cdot (3k-|T|-|R|)_{|S|}$.
We now prove property $(c)$, that is more involved.
First of all, observe that for every $H\in \mathcal{F}(T,\gamma)$ the value of $\pseudo_T(y_H)$ depends only on $T$ and the configuration profile $\gamma$.
More specifically,
\[\pseudo_T(y_H)=\frac{1}{(3k-|T|)}_{\|\gamma\|}\prod_{j=1}^{6}(k/2-\delta_T(C_j))_{\gamma(C_j)},\]
since $|H|=\|\gamma\|$ and $\delta_H(C_j)=\gamma(C_j)$ for every $j\in \{1,\ldots,6\}$.
Then, $\pseudo_T(\mathcal{B}_{T,\gamma})$ equals $|\mathcal{F}(T,\gamma)|$ times the quantity above.
The number of machines that  can support a partial schedule $H$ that extend $T$ is $3k-|T|$, and since $|H|=\|\gamma\|$ the number of possible machine domains is
${3k-|T| \choose \|\gamma\|}.$
Given a set of machines with cardinality $\|\gamma\|$, the number of partial schedules with domain equal to this set of machines and that are in configuration profile $\gamma$ are
$\|\gamma\| ! \prod_{j=1}^{6}\frac{1}{\gamma(C_j)!}.$
Then, overall, the value of $\pseudo_T(\mathcal{B}_{T,\gamma})$ is equal to
\begin{align*}
&{3k-|T| \choose \| \gamma\|} \|\gamma\| ! \frac{1}{(3k-|T|)}_{\|\gamma\|} \cdot \prod_{j=1}^{6}\frac{1}{\gamma(C_j)!}(k/2-\delta_T(C_j))_{\gamma(C_j)}\\
=&\;\frac{(3k-|T|)!}{(3k-|T|-\|\gamma\|)!}\cdot \frac{1}{(3k-|T|)}_{\|\gamma\|}\cdot \prod_{j=1}^{6}\frac{1}{\gamma(C_j)!}(k/2-\delta_T(C_j))_{\gamma(C_j)}\\
=&\;\prod_{j=1}^{6}\frac{1}{\gamma(C_j)!}(k/2-\delta_T(C_j))_{\gamma(C_j)},
\end{align*}
in the last step we used that for every real $x$ and non-negative integer $b$, it holds $(x-b)!(x)_b=x!$.
 \end{proof}

To prove Lemma~\ref{lem:key-lem} we obtain first a weaker version, that together with a polynomial decomposition in the scheduling ideal yields to the pseudoindependence result.

\begin{lemma}
\label{lem:weaker}
Let $T\subseteq [m]\times \CC$ be a partial schedule.
\begin{enumerate}
	\item[$(a)$] If $\nu$ and $\xi$ are configuration profiles such that $\text{supp}(\nu)\cap\text{supp}(\xi)=\emptyset$ and $|T|+\|\nu\|+\|\xi\|\le k/2$, then $\pseudo_T(\mathcal{B}_{T,\nu}\mathcal{B}_{T,\xi})=\pseudo_T(\mathcal{B}_{T,\nu})\cdot \pseudo_T(\mathcal{B}_{T,\xi}).$
	\item[$(b)$] If $\nu$ and $\xi$ are configuration profiles such that there exists $C\in \{C_1,\ldots,C_6\}$ with $\text{supp}(\nu),\text{supp}(\xi)\subseteq \{C\}$, and $|T|+\|\nu\|+\|\xi\|\le k/2$, then we have that $\pseudo_{T}(\mathcal{B}_{T,\nu}\mathcal{B}_{T,\xi})=\pseudo_T(\mathcal{B}_{T,\nu})\cdot \pseudo_T(\mathcal{B}_{T,\xi})$.
\end{enumerate}
\end{lemma}

\begin{proof}
In both case if one of the profiles is zero then the conclusion follows.
Then, in what follows assume that $\nu$ and $\xi$ are different from zero, and their support is contained in $\{C_1,\ldots,C_6\}$.
Consider $\nu$ and $\xi$ satisfying the conditions in $(a)$ and fix $A\in \mathcal{F}(T,\nu)$.
Then,
\begin{align*}
\pseudo_T(y_A\mathcal{B}_{T,\xi})
=&\sum_{B\in \mathcal{F}(T\cup A,\xi)}\pseudo(y_Ay_B)+\sum_{B\in \mathcal{F}(T,\xi)\setminus \mathcal{F}(T\cup A,\xi)}\pseudo(y_Ay_B),
\end{align*}
where the equality holds since $\mathcal{F}(T\cup A,\xi)\subseteq \mathcal{F}(T,\xi)$.
For every term $B\in \mathcal{F}(T,\xi)\setminus \mathcal{F}(T\cup A,\xi)$ we have that it is incident to at least one of the machines in $G_A$.
Since every machine in $G_A$ is connected to a machine in $\text{supp}(\nu)\subseteq \CC\setminus \text{supp}(\xi)$, it follows that $A\cup B$ is not a partial schedule since at least one machine is connected to different configurations, and in consequence its pseudoexpectation is zero.
Therefore, the second summation in the equality above is zero. 
Together with property (b) in Lemma~\ref{lem:conditioning-b} it implies that
\begin{align*}
\pseudo_T(y_A\mathcal{B}_{T,\xi})=\sum_{B\in \mathcal{F}(T\cup A,\xi)}\pseudo(y_Ay_B)=\pseudo_T(y_A)\cdot \pseudo_{T\cup A}(\mathcal{B}_{T\cup A,\xi}).
\end{align*}
Since $\text{supp}(\nu)\cap \text{supp}(\xi)=\emptyset$, we have that for every $C_j\in \text{supp}(\xi)$, $\delta_{T\cup A}(C_j)=\delta_T(C_j)$.
On the other hand, if $C_j\notin \text{supp}(\xi)$ then $(x)_{\xi(C_j)}=(x)_0=1$ for every real $x$.
Overall, and together with Lemma~\ref{lem:conditioning-d}, it holds that
\begin{align*}
\pseudo_{T\cup A}(\mathcal{B}_{T\cup A,\xi})&=\prod_{j=1}^6\frac{1}{\xi(C_j)!}(k/2-\delta_{T\cup A}(C_j))_{\xi(C_j)}\\
&=\prod_{j\in \text{supp}(\xi)}\frac{1}{\xi(C_j)!}(k/2-\delta_{T}(C_j))_{\xi(C_j)}=\pseudo_T(\mathcal{B}_{T,\xi}).
\end{align*}
Together with the linearity of $\pseudo_T$ we conclude $(a)$.
Consider now $\nu,\xi$ satisfying the conditions in $(b)$, and let $C\in \{C_1,\ldots,C_6\}$ the configuration that supports both profiles.
Without loss of generality suppose that $\nu(C)\ge \xi(C)$.
For $A\in \mathcal{F}(T,\nu)$ and $B\in \mathcal{F}(T,\xi)$, we have that $A\cup B$ is always a perfect matching since the profiles are supported in the same configuration.
If $B\subseteq A$, then the union has profile $\nu$.
Then, by Lemma~\ref{lem:conditioning} $(c)$ we have
\begin{align*}
\pseudo_T(y_A\mathcal{B}_{T,\xi})&=\sum_{B\in \mathcal{F}(T,\xi)}\pseudo_T(y_Ay_B)\\
&=\sum_{B\in \mathcal{F}(T,\xi):B\subseteq A}\pseudo_T(y_A)+\sum_{B\in \mathcal{F}(T,\xi):B\setminus A\ne \emptyset}\pseudo_T(y_Ay_{B\setminus A})\\
&=\pseudo_T(y_A)\left({\nu(C)\choose \xi(C)}+\sum_{B\in \mathcal{F}(T,\xi):B\setminus A\ne \emptyset}\pseudo_{T\cup A}(y_{B\setminus A})\right).
\end{align*}
If $B\setminus A\ne \emptyset$, the union profile can be parameterized in $|B\setminus A|=\omega$, and let $\alpha_{\omega}$ be the profile such that $\alpha_{\omega}(C)=\omega$ and zero otherwise.
Thus,
\begin{align*}
&\sum_{B\in \mathcal{F}(T,\xi):B\setminus A\ne \emptyset}\pseudo_{T\cup A}(y_{B\setminus A})\\
&=\sum_{\omega=1}^{\xi(C)}{\nu(C) \choose \xi(C)-\omega}{3k-|T|-\nu(C) \choose \omega}\frac{(k/2-\delta_T(C)-\nu(C))_{\omega}}{(3k-|T|-\nu(C))_{\omega}}\\
																		&=\sum_{\omega=1}^{\xi(C)}\frac{1}{\omega!}{\nu(C) \choose \xi(C)-\omega}(k/2-\delta_T(C)-\nu(C))_{\omega},
\end{align*}
and since $(k/2-\delta_T(C)-\nu(C))_{0}=1$, and running the summation over $A\in \mathcal{F}(T,\nu)$ we obtain over all that
\begin{equation}
\label{eq:casicasi}
\pseudo_T(\mathcal{B}_{T,\nu}\mathcal{B}_{T,\xi})=\pseudo_{T}(\mathcal{B}_{T,\nu})\cdot \sum_{\omega=0}^{\xi(C)}\frac{1}{\omega!}{\nu(C) \choose \xi(C)-\omega}(k/2-\delta_T(C)-\nu(C))_{\omega}.
\end{equation}
\begin{claim}
\label{claim:stirling}
Let $a$ and $b$ be two non-negative integers such that $a\le b$. 
Then, for every real $x$,
\[\sum_{\omega=0}^{a}\frac{1}{\omega!}{b \choose a-\omega}(x-b)_{\omega}=\frac{1}{a!}(x)_{a}.\]
\end{claim}
\noindent The claim applied in (\ref{eq:casicasi}) for $x=k/2-\delta_T(C)$, $a=\xi(C)$ and $b=\nu(C)$ yields the result, since
\[\pseudo_T(\mathcal{B}_{T,\nu}\mathcal{B}_{T,\xi})=\pseudo_{T}(\mathcal{B}_{T,\nu})\cdot \frac{1}{\xi(C)!}(k/2-\delta_T(C))_{\xi(C)}=\pseudo_{T}(\mathcal{B}_{T,\nu})\pseudo_{T}(\mathcal{B}_{T,\xi}).\]
The claim follows by the Chu-Vandermonde identity~\cite[p. 59-60]{askey75}, 
\begin{align*}
(x)_a&=\sum_{\omega=0}^{a}{a\choose \omega}(x-b)_{\omega}(b)_{a-\omega}=a!\sum_{\omega=0}^{a}(x-b)_{\omega}\frac{(b)_{a-\omega}}{(a-\omega)!}=a!\sum_{\omega=0}^{a}(x-b)_{\omega}{b\choose a-\omega}\qedhere
\end{align*}
 \end{proof}

\begin{proof}[Lemma~\ref{lem:key-lem}]
Given a profile configuration $\gamma$ and $C_j\in \{C_1,\ldots,C_6\}$, we denote by $\gamma_j$ the profile that is zero for every $C\ne C_j$ and $\gamma_j(C_j)=\gamma(C_j)$.
In the following, we prove that the following factorization
holds: 
\begin{equation}
\label{eq:decomposition}
\pseudo_T(\mathcal{B}_{T,\gamma}\mathcal{B}_{T,\mu})=\pseudo_T\left(\prod_{j=1}^{6}\mathcal{B}_{T,\gamma_j}\mathcal{B}_{T,\mu_j}\right),
\end{equation}
recalling that $\mathcal{B}_{T,\xi}=1$ if $\xi=0$.
Before checking that the decomposition above is correct, we see how to conclude the lemma from that.
Observe that by construction $\text{supp}(\gamma_j)\cap \text{supp}(\gamma_{\ell})=\emptyset$ if $j\ne \ell$, and therefore by Lemma~\ref{lem:weaker} $(a)$, we have
\begin{equation}
\pseudo_T\left(\prod_{j=1}^{6}\mathcal{B}_{T,\gamma_j}\mathcal{B}_{T,\mu_j}\right)=\prod_{j=1}^{6}\pseudo_T\left(\mathcal{B}_{T,\gamma_j}\mathcal{B}_{T,\mu_j}\right).
\end{equation}
Furthermore, since for every $j\in \{1,\ldots,6\}$ we have $\text{supp}(\gamma_j),\text{supp}(\mu_j)\subseteq \{C_j\}$, by Lemma~\ref{lem:weaker} $(b)$ we have 
\[\prod_{j=1}^{6}\pseudo_T\left(\mathcal{B}_{T,\gamma_j}\mathcal{B}_{T,\mu_j}\right)=\prod_{j=1}^{6}\pseudo_T(\mathcal{B}_{T,\gamma_j})\pseudo_T(\mathcal{B}_{T,\mu_j})=\prod_{j=1}^{6}\pseudo_T(\mathcal{B}_{T,\gamma_j})\cdot \prod_{j=1}^{6}\pseudo_T(\mathcal{B}_{T,\mu_j}).\]
By using Lemma~\ref{lem:weaker} $(a)$ 
the right hand side is equal to
\begin{align*}
														\pseudo_T\left(\prod_{j=1}^{6}\mathcal{B}_{T,\gamma_j}\right)\cdot \pseudo_T\left(\prod_{j=1}^{6}\mathcal{B}_{T,\mu_j}\right)=\pseudo_T(\mathcal{B}_{T,\gamma}\mathcal{B}_{T,\mu}),
\end{align*}
where in the last equality we used the decomposition in (\ref{eq:decomposition}) separately for $\gamma$ and $\mu$.
We check now that the factorization in (\ref{eq:decomposition}) is always valid.
Let $S$ be a partial schedule disjoint from $T$ and with profile $\mu$ and let $C_j\in \text{supp}(\gamma)$. 
It is enough to check that 
\begin{equation}
\pseudo_T(\mathcal{B}_{T,\gamma}y_S)=\pseudo_T(\mathcal{B}_{T,\gamma_j}\mathcal{B}_{T,\gamma-\gamma_j}y_S),
\end{equation}
since the factorization follows by the linearity of $\pseudo_T$ and by applying iteratively for every $C_j\in \{C_1,\ldots,C_6\}$ the above factorization.
We have that
\begin{align*}
\pseudo_T(\mathcal{B}_{T,\gamma}y_S)&=\pseudo_T\left(\sum_{A\in \mathcal{F}(T,\gamma)}y_Ay_S\right)\\
&=\pseudo_T\left(\sum_{B\in \mathcal{F}(T,\gamma_j)}y_B\sum_{D\in \mathcal{F}(T\cup B,\gamma-\gamma_j)}y_Dy_S\right).
\end{align*}
Fix $B\in \mathcal{F}(T,\gamma_j)$ and consider a set $D\in \mathcal{F}(T,\gamma-\gamma_j)\setminus \mathcal{F}(T\cup B,\gamma-\gamma_j)$.
In particular, $D$ is in profile $\gamma-\gamma_j$ but is incident to at least one machine, say $\ell$, that is also incident to $B$.
Since $B$ is in profile $\gamma_j$ and it has disjoint support from $\gamma-\gamma_j$, the above implies that machine $\ell$ is incident to different configurations, and therefore its pseudoexpectation value is equal to zero. 
That is the contribution to the pseudoexpectation value of the terms in $\mathcal{F}(T,\gamma-\gamma_j)\setminus \mathcal{F}(T\cup B,\gamma-\gamma_j)$ is is zero.
Furthermore, since $\mathcal{F}(T,\gamma-\gamma_j)\supseteq \mathcal{F}(T\cup B,\gamma-\gamma_j)$, we have that for every $B\in \mathcal{F}(T,\gamma_j)$, 
\begin{align*}
&\pseudo_T\left(y_B\sum_{D\in \mathcal{F}(T\cup B,\gamma-\gamma_j)}y_Dy_S\right)\\
&=\pseudo_T\left(y_B\left(\sum_{D\in \mathcal{F}(T\cup B,\gamma-\gamma_j)}y_D+\sum_{D\in \mathcal{F}(T,\gamma-\gamma_j)\setminus \mathcal{F}(T\cup B,\gamma-\gamma_j)}y_D\right)y_S\right)\\
&=\pseudo_T\left(y_B\sum_{D\in \mathcal{F}(T,\gamma-\gamma_j)}y_Dy_S\right)=\pseudo_T(y_B\mathcal{B}_{T,\gamma-\gamma_j}y_S).
\end{align*}
We conclude by summing over $B\in \mathcal{F}(T,\gamma_j)$, $S\in \mathcal{F}(T,\mu)$ and using the linearity of $\pseudo_T$.
 \end{proof}

\noindent{\textbf{Remark.}} It is worth noticing that the lower bound of Theorem~\ref{thm:negative} translates to the weaker assignment linear program (see~\eqref{eq:assign_jobs}-\eqref{eq:assign_noneg} in Section~\ref{sec:upper-bound}, which define the linear program $\ass{}(T) $). More precisely, there exists an instance such that, after applying $\Omega(n)$ rounds of the \sos{} hierarchy to the assignment linear program, the semidefinite relaxation has an integrality gap of at least 1.0009. This follows by Theorem~\ref{thm:negative} and a general result by Au and Tun\c{c}el~\cite{au_elementary_2018}.
More details can be found in Appendix~\ref{sec:appendixB}.

\section{Upper Bound: Breaking Symmetries to Approximate the Makespan}
\label{sec:upper-bound}
\label{sec:sym-breaking} In the previous section we showed that the configuration linear program has an inherent difficulty for the \sos{} method with low (constant) degree to yield a $(1+\varepsilon)$ integrality gap (and hence also the weaker assignment linear program below). It is natural to ask whether there is a way to avoid this lower bound. As suggested by our proof in Section~\ref{sec:lb} and several other lower bounds in the literature~\cite{laurent2003lower,grigoriev2001,potechin17,KMMMVW18,RSS18}, symmetries seem to play an role in the quality of the relaxations obtained by the \sos{} and SA hierarchies. A natural question is whether \emph{breaking} the symmetries of a problem or instance might help avoiding the lower bounds. In what follows we show that this is the case for the makespan scheduling problem. We leave as an interesting open problem whether this is the case for other relevant problems.\\

\noindent{\it Symmetry Breaking.} Breaking symmetries is a common technique to avoid algorithmic problems of symmetric instances of non-convex programs, in particular integer programs~\cite{margot_symmetry_2010}. Recall that given an optimization problem (P): $\min\{f(x): x\in X\}$ for some set $X\subseteq \mathbb{R}^n$ and a group $G$ acting on $\mathbb{R}^n$ by an action $(g,x)\mapsto gx$, we say that (P) is $G$-invariant if $f(x)=f(gx)$ and $gx\in X$ for all $x\in X$ and $g\in G$. Notice that if $x^*$ is an optimal solution to (P), then $gx^*$ is also optimal for every $g\in G$ in this case. Hence, if we add to the formulation any inequality $a^{\top}x\le b$ that keeps at least one representative of any given orbit $\{gx:g\in G\}$ for any $x\in \mathbb{R}^n$, that is, for all $x\in \mathbb{R}^n$ there exists $g\in G$ such that $a^{\top}(gx)\le b$, then we guarantee that (P'): $\min\{f(x): x\in X, a^{\top}x\le b\}$ contains at least one optimal solution. If such inequality is not valid for (P), we say that it is a \emph{symmetry breaking inequality}.\footnote{It is worth noticing that such an inequality might not break all symmetries, that is, we do not require that there is a unique representative of each orbit.}\\

\noindent{\it Application to Scheduling.} We show that we can obtain almost optimal relaxations
in terms of the integrality gap if we add a well chosen set of symmetry breaking inequalities to a ground formulation and then apply the \SA{} hierarchy (which is even weaker than \sos{}).
Furthermore, the ground formulation we use is the \emph{assignment linear program}.
In this LP there are variables $x_{ij}$ indicating whether job $j$
is assigned to machine~$i$. 
For an estimate or guess $T$ for the optimal makespan we denote by
$\ass(T)$ the formulation given by 
\begin{alignat}{2}
\label{eq:assign_jobs}
\sum_{i\in [m]}x_{ij} & =1 &  & \quad\text{for all}\;j\in J,\\
\sum_{j\in J}x_{ij}p_{j} & \leq T &  & \quad\text{for all}\;i\in [m],\\
\label{eq:assign_noneg}
x_{ij} & \ge0 &  & \quad\text{for all}\;i\in [m],\text{ for all }j\in J.
\end{alignat}
If we require that $T\ge\max_{j\in J}p_{j}$ then
the assignment linear program has an integrality gap of 2~\cite{williamson2011design}.\\

\noindent{\it Roadmap.} In Subsection~\ref{subsec:symmetry-break-ineq} we define the symmetry breaking inequalities that we will add to the assignment linear program. In Subsection~\ref{sec:Integrality-gap-SA} we will show how to round a feasible solution of the \SA{} hierarchy with $2^{\tilde{O}(1/\varepsilon^2)} $ rounds over this program to obtain an integral solution with makespan $(1+\varepsilon)T$. In Subsection~\ref{sec:faster-ptas} we will show that breaking some new \emph{approximate symmetries}, with just $O(1/\varepsilon^5)$ rounds of the \SA{} hierarchy suffices to obtain a $(1+\varepsilon)$-approximate solution, yielding an exponential decrease in the number of necessary rounds. By approximate symmetries we mean that first we round similar processing times to the same value, and then add symmetry breaking inequalities for the new induced symmetries.

\subsection{Symmetry breaking inequalities}
\label{subsec:symmetry-break-ineq}
In order to define our symmetry breaking inequalities we consider a partitioning obtained by grouping long jobs with a \textit{\emph{similar}}\emph{
}processing time. Let $\varepsilon\in (0,1)$ such that $1/\varepsilon\in \mathbb{Z}$. We say that a job $j\in J$ is \emph{long} if
$p_{j}\ge\varepsilon\cdot T$, and it is \emph{short} otherwise. The
subset of long jobs is denoted by $J_{\largo}$ and the short jobs
are $J_{\corto}=J\setminus J_{\largo}$. 
For every $q\in\{1,\ldots,(1-\varepsilon)/\varepsilon^{2}\}$
we define 
\[
J_{q}=\left\{ j\in J_{\largo}:\left(\frac{1}{\varepsilon}+q\right)\varepsilon^{2}T>p_{j}\ge\left(\frac{1}{\varepsilon}+q-1\right)\varepsilon^{2}T\right\} .
\]
Let $s:=(1-\varepsilon)/\varepsilon^{2}$ denote the number of groups
of long jobs. The reader may imagine that for each group $J_{q}$
with $q\in[s]$ we round the size of each job $j\in J_{q}$ to $\left(\frac{1}{\varepsilon}+q\right)\varepsilon^{2}T$.
This increases the overall makespan at most by a factor $1+\varepsilon$.
Also note that if we can find a schedule for the long jobs with makespan
at most $(1+\varepsilon)T$, then there is also a schedule for \emph{all}
jobs with makespan at most $(1+\varepsilon)T$ since we can add the short
jobs in a greedy manner (see e.g.,\cite{williamson2011design}; 
we assume that $\ass(T)$ is feasible and then $\sum_{j\in J}p_j \le m\cdot T$ holds).\\

\noindent{\it Configurations.} Based on the partition of the long jobs $\{J_{q}\}_{q\in[s]}$ we
define configurations of the long jobs. We say that a \textit{configuration}
$C$ is a multiset of elements in $\{1,\ldots,s\}$. Let $\CC$ denote
the set of all configurations. Similarly as in Section~\ref{sec:lb},
for a configuration $C$ we define $m(q,C)$ to be the number of times
that $q$ appears (repeated) in $C$. Intuitively, this means that
configuration $C$ contains $m(q,C)$ slots for jobs in $J_{q}$.
In what follows, we introduce a set of constraints that guarantees
\textit{\emph{that }}\textit{every} integer solution to $\ass(T)$
obeys a specific order on the configurations over the machines, i.e.,
there is a total order of the configurations $C$ such that for two
machines $i,i'\in [m]$ with $i<i'$ the configuration on $i$ is smaller
according to this total ordering than the configuration on $i'$.
This is a way of \textit{breaking} the symmetries due to permuting
machines.
Formally, we say that a configuration $C$ is \textit{lexicographically}
larger than a configuration $C'$ if there exists $q\in[s]$ such
that $m(\ell,C)=m(\ell,C')$ for all $\ell<q$ and $m(q,C)>m(q,C')$.
We denote this by $C>_{\lex}C'$. In particular, the relation $>_{\lex}$
defines a total order over $\CC$. \\

\noindent{\it Integer linear program.} Let $B:=1+2s\max_{q\in[s]}|J_{q}|=O(|J|^{2})$. We define an integer
linear program $\ass(B,T)$ below in which we enforce that the machines
are ordered according to the relation $>_{\lex}$.
\begin{alignat}{2}
\sum_{i\in [m]}x_{ij} & =1 &  & \quad\text{for all}\;j\in J,\\
\sum_{j\in J}x_{ij}p_{j} & \leq T &  & \quad\text{for all}\;i\in [m],\\
\sum_{q=1}^{s}B^{s-q}\sum_{j\in J_{q}}\left(x_{ij}-x_{(i+1)j}\right) & \ge0 &  & \quad\text{for all}\;i\in[m-1],\\
x_{ij} & \ge0 &  & \quad\text{for all}\;i\in [m],\text{ for all }j\in J.
\end{alignat}
To avoid confusion we sometimes use the notation $\ass(J,B,T)$ to
emphasize that we are considering the program for the job set $J$.
Given a subset of jobs $K\subseteq J$ such that $\sum_{j\in K}p_{j}\le T$,
we denote by $\text{conf}(K)$ the configuration such that for every
$q\in\{1,\ldots,s\}$, $m(q,\config(K))=|K\cap J_{q}|$. We then say
that $\config(K)$ is the \textit{configuration induced by $K$}.
In the following we show that every integer solution to the program
$\ass(B,T)$ obeys the lexicographic order $>_{\lex}$ on the configurations
over the machines. More specifically, given a feasible integer solution
$x\in\ass(T)$ and a machine $i\in [m]$, let $\conf_{i}(x)\in\CC$
be the configuration defined by the job assignment of $x$ to machine
$i$, that is, for every $q\in\{1,\ldots,s\}$, $m(q,\conf_{i}(x))=\sum_{j\in J_{q}}x_{ij}$.

\begin{theorem}\label{lem:lexico-mach} In every integer solution
$x\in\ass(B,T)$, for every machine $i\in [m-1]$ we have
that $\conf_{i}(x)\ge_{\lex}\conf_{i+1}(x)$.
\end{theorem}
To prove Theorem~\ref{lem:lexico-mach}, we define $\mathcal{L}_{B}:\CC\to\RR$
to be the function such that for every configuration $C\in\CC$, $\mathcal{L}_{B}(C)=\sum_{q=1}^{s}B^{s-q}m(q,C).$
The important point is that $\mathcal{L}_{B}$ is strictly increasing.

\begin{lemma}\label{lem:lex-increase} For two configurations $C,C'\in\CC$
with $C<_{\lex}C'$ we have that $\mathcal{L}_{B}(C)<\mathcal{L}_{B}(C')$.
\end{lemma}

\begin{proof}Consider two configurations $C,C'\in\CC$ such that
$C<_{\lex}C'$. Let $\tilde{q}$ be the smallest integer in $\{1,\ldots,s\}$
such that the multiplicities of the configurations are different,
that is, $m(\ell,C)=m(\ell,C')$ for every $\ell<\tilde{q}$. Hence
it holds that $m({\tilde{q}},C)<m({\tilde{q}},C')$. In particular,
every term up to $\max\{0,\tilde{q}-1\}$ in the summation defining
$\mathcal{L}_{B}(C)-\mathcal{L}_{B}(C')$ is equal to zero. By upper
bounding the summation from $\min\{s,\tilde{q}+1\}$ we get that
$\sum_{q=\min\{s,\tilde{q}+1\}}^{s}B^{s-q}\left(m(q,C)-m(q,C')\right)$
is at most 
\begin{align*}
 & \sum_{q=\min\{s,\tilde{q}+1\}}^{s}B^{s-q}\left(|m(q,C)|+|m(q,C')|\right)\le\sum_{q=\min\{s,\tilde{q}+1\}}^{s}B^{s-q}\cdot2|J_{q}|<B^{*}\cdot B^{s-\tilde{q}-1}<B^{s-\tilde{q}},
\end{align*}
and since $m({\tilde{q}},C')-m({\tilde{q}},C)\ge1$ it follows that
\begin{align*}
\sum_{q=\tilde{q}}^{s}B^{s-q}\left(m(q,C)-m(q,C')\right) & <B^{s-\tilde{q}}\left(m({\tilde{q}},C)-m({\tilde{q}},C')\right)+B^{s-\tilde{q}}\\
 & <B^{s-\tilde{q}}\left(m({\tilde{q}},C)-m({\tilde{q}},C')+1\right)<0
\end{align*}
and hence $\mathcal{L}_{B}(C)<\mathcal{L}_{B}(C')$.  \end{proof}

\begin{proof}[Theorem~\ref{lem:lexico-mach}] Fix a machine
$i\in [m-1]$. Since $x$ is an integral solution in $\ass(B,T)$,
we have that $\conf_{i}(x),\conf_{i+1}(x)\in\CC$. The symmetry breaking
constraints implies that 
\begin{align*}
0 & \le\sum_{q=1}^{s}B^{s-q}\left(m(q,\conf_{i}(x))-m(q,\conf_{i+1}(x))\right)=\mathcal{L}_{B}(\conf_{i}(x))-\mathcal{L}_{B}(\conf_{i+1}(x)).
\end{align*}
Applying Lemma \ref{lem:lex-increase} it holds that $\mathcal{L}_{B}$
is strictly increasing and therefore $\conf_{i}(x)\ge_{\lex}\conf_{i+1}(x)$.
 \end{proof}


In general, $\ass(B,T)$ is \textit{not} $S_{m}$-invariant, that is,
given a solution to $\ass(B,T)$, if we permute the machines then
we do not necessarily obtain another solution for it. However, it
is a valid formulation, in the sense that if there exists a schedule
with makespan at most $T$, then $\ass(B,T)$ has a feasible integral
solution (more precisely, we retain a representative solution for
each orbit). To show this, we can take an arbitrary schedule of makespan
$T$ and reorder the machines lexicographically according to their
configurations.

\begin{lemma}\label{lem:valid-form} If there exists an integral
feasible solution to $\ass(T)$ then there exists also an integral
feasible solution to $\ass(B,T)$. \end{lemma}

\begin{proof} 
Since there exists a schedule of makespan at most $T$,
there exists an integral solution $x\in\ass(T)$. Since the lexicographic
relation defines a total order over $\CC$, there exists a permutation
$\sigma\in S_{m}$ such that for every $i\in [m-1]$, $\conf_{\sigma(i)}(x)\ge_{\lex}\conf_{\sigma(i+1)}(x)$.
Consider the integral solution $\tilde{x}$ obtained by permuting
the solution according to $\sigma$, that is, $\tilde{x}=\sigma x$.
Then, for every $i\in [m-1]$ it follows that $\sum_{q=1}^{s}B^{s-q}\sum_{j\in J_{q}}\left(\tilde{x}_{ij}-\tilde{x}_{(i+1)j}\right)$
is equal to 
\begin{align*}
 & \sum_{q=1}^{s}B^{s-q}\left(m(q,\conf_{\sigma(i)}(x))-m(q,\conf_{\sigma(i+1)}(x))\right) =\mathcal{L}_{B}(\conf_{\sigma(i)}(x))-\mathcal{L}_{B}(\conf_{\sigma(i+1)}(x))\ge0.
\end{align*}
The last step holds by Lemma~\ref{lem:lex-increase}. We conclude
that $\tilde{x}\in\ass(B,T)$.  \end{proof} 


\subsection{LP based approximation scheme}

\label{sec:Integrality-gap-SA}

In this section we prove Theorem~\ref{thm:positive}, i.e., we show
that if we apply $2^{\tilde{O}(1/\varepsilon^{2})}$ rounds of the
Sherali-Adams hierarchy to $\ass(B,T)$ then the integrality gap of
the resulting LP is at most $1+\epsilon$,
i.e., if it has a feasible solution then there exists an integral solution with makespan at most $(1+\varepsilon)T$. 
Recall the definition of
a $\SA$ pseudoexpectation at the end of Section~\ref{sec:prelim};
in particular, recall that if a degree-$r$ $\SA$ pseudoexpectation
exists for a linear program, then it has a solution after applying
$r$ rounds of $\SA$ to it. The main result of this section is the
following theorem. 

\begin{theorem} \label{thm:int-gap-SA} Consider a value $T>0$ and
suppose there exists a degree-$\nrounds$ $\SA$ pseudoexpectation
for $\assign(B,T)$. Then, there exists an integral solution in $\assign(B,(1+\varepsilon)T)$
and it can be computed in polynomial time. 
\end{theorem}

In what follows we might omit $\SA$ when referring to pseudoexpectations
since the context is clear. %
Given a degree-$r$ pseudoexpectation $\pseudo$ and a subset $A\subseteq [m]\times J$
with $\pseudo(x_{A})\neq0$, we define the $A$-\textit{conditioning}
to be the linear operator over $\RR[x]/\I_{E}$ defined by 
$\pseudo_{A}(x_{I})=\pseudo(x_{I}x_{A})/\pseudo(x_{A})$,
for every $I\subseteq [m]\times J$. We also say that \emph{we condition
on $A$.} The following lemma summarises some of the relevant properties
of the conditionings. We refer to~\cite{Lau03} for a proof of it
as well as a detailed exposition of the $\SA$ hierarchy. 
\begin{lemma} \label{lem:SA-properties} Let $\pseudo$ be a degree-$r$
pseudoexpectation and let $\pseudo_{A}$ be the conditioning for some
$A\subseteq [m]\times J$ of cardinality at most $r$. 
Then 
\begin{enumerate}
\item[$(a)$] $\pseudo_{A}(x_{A})=1$ and $\pseudo_{A}(x_{ij})=1$ for every $(i,j)\in A$. 
\item[$(b)$] $\pseudo_{A}$ is a degree-$(r-|A|)$ pseudoexpectation. 
 
\item[$(c)$] For every $B\subseteq [m]\times J$ such that $\pseudo(x_{B})\in\{0,1\}$
we have $\pseudo(x_{B})=\pseudo_{A}(x_{B})$. 
\end{enumerate}
\end{lemma}


\noindent{\bf Stability.} In the following consider a degree-$r$ pseudoexpectation
$\pseudo$ for $\ass(J_{\largo},B^{*},T)$, and let $\{J_{1},\ldots,J_{s}\}$
be the partitioning of $J_{\largo}$ defined above. Recall that $s=(1-\varepsilon)/\varepsilon^{2}$.
Our strategy is to find a set $A\subseteq [m]\times J$ with $|A|\le2^{\tilde{O}(1/\varepsilon^{2})}$
such that in $\pseudo_{A}$ for each machine $i$ and for each set
$J_{q}\in\{J_{1},\ldots,J_{s}\}$ an \emph{integral} number of jobs
from $J_{q}$ are assigned to $i$. Since in each set $J_{q}$ the
jobs have essentially the same length, based on $\pseudo_{A}$ we
can compute an assignment of the long jobs to the machines of makespan
at most $(1+\epsilon)T$. In order to find the set $A$, we will apply
Lemma~\ref{lem:SA-properties} several times. In the process we will
achieve that for some machines $i$ the number of jobs from some set
$J_{q}$ is integral and does not change if we condition on further
elements from $[m]\times J$. Formally, we define that for some $q\in\{1,\ldots,s\}$
and $a\in\N$ a machine $i$ is \emph{$(q,a)$-stable in }$\pseudo$
if we have 
\[\sum_{j\in J_{q}}\pseudo(x_{ij})=a \;\text{ and } 
\;\sum_{j\in J_{q}}\pseudo_{A}(x_{ij})=a \;\text{ for any $A$-conditioning}\]
with $A\subseteq [m]\times J$. 
We say that a machine $i\in [m]$ is $q$\emph{-stable in }$\pseudo$
if it is $(q,a)$-stable for some $a\in\ZZ_{+}$. We will apply the
following lemma over $\pseudo$ several times, until each machine
$i$ is $q$-stable for each $q\in\{1,...,s\}$.

\begin{lemma} \label{lem:conditioning-levels} Consider $q\in\{1,\ldots,s\}$,
integers $a_{1},...,a_{q}$, a degree-$r$ pseudoexpectation $\pseudo$
and a set of consecutive machines $\{i_{L},...,i_{R}\}$, with $r\ge1/\varepsilon^{2}$.
Suppose that every machine $i\in\{i_{L},\ldots,i_{R}\}$ is $(\tilde{q},a_{\tilde{q}})$-stable
in $\pseudo$ for each $\tilde{q}\in\{1,...,q\}$. Then, there is
a degree-$(r-2/\varepsilon^{2})$ pseudoexpectation $\pseudo^{\mathsf{stab}}$
such that each machine $i\in\{i_{L},\ldots,i_{R}\}$ is $\hat{q}$-stable
in $\pseudo^{\mathsf{stab}}$ for each $\hat{q}\in\{1,...,q+1\}$.
\end{lemma}

\paragraph{Phase 1: Obtaining a good pseudoexpectation.}

\noindent We first use Lemma~\ref{lem:conditioning-levels} in order
to prove Theorem~\ref{thm:int-gap-SA}. We prove Lemma~\ref{lem:conditioning-levels}
later in Section~\ref{subsec:pf-lemma-cond}. In what follows let
$\pseudo$ be a degree-$\nrounds$ pseudoexpectation for $\assign(B,T)$.
Our algorithm works in $s$ \emph{stages}. After stage $q$ we obtain
a pseudoexpectation in which each machine is $\tilde{q}$-stable for
each $\tilde{q}\in\{1,\ldots,q\}$. In the first stage we apply Lemma~\ref{lem:conditioning-levels}
on the solution $\pseudo$ with $i_{L}=1$, $i_{R}=m$ and $q=0$,
and let $\pseudo^{(0)}$ be the pseudoexpectation obtained. Assume by
induction that after stage $q$ we have obtained a pseudoexpectation
$\pseudo^{(q)}$ in which each machine is $\tilde{q}$-stable for
$\tilde{q}\in\{1,\ldots,q\}$. Consider a partition of $[m]$ given
by $\{M_{1},\ldots,M_{k}\}$ of the machines such that in $\pseudo^{(q)}$,
for each set $M_{\ell}$ with $\ell\in\{1,\ldots,k\}$, there
are integers $a_{\ell,1},\ldots,a_{\ell,q}$ such that each machine
in $M_{\ell}$ is $(\tilde{q},a_{\ell,\tilde{q}})$-stable for
each $\tilde{q}\in\{1,\ldots,q\}$. Since $\pseudo^{(q)}$ is a
pseudoexpectation for $\assign(B,T)$, which includes the symmetry
breaking constraints, 
the machines in each set $M_{\ell}$ are consecutive. Since the possible
number of combinations $a_{\ell,1},\ldots,a_{\ell,q}$ is at most
$(1/\varepsilon+1)^{q}$ we can find such a partition with $k\le(1/\varepsilon+1)^{q}$.
For each $\ell\in\{1,\ldots,k\}$ we apply Lemma~\ref{lem:conditioning-levels}.
Hence, the total number of rounds in this stage is at most $(1/\varepsilon+1)^{q}\cdot2/\varepsilon^{2}$.
Denote by $\pseudo^{(q+1)}$ the obtained solution. We continue for
$s$ stages. Let $\pseudo^{\text{f}}$ be the pseudoexpectation returned
by the algorithm. The degree of $\pseudo^{\text{f}}$ is at least
$\nrounds-\sum_{q=1}^{s}(1/\varepsilon+1)^{q}\cdot2/\varepsilon^{2}\ge0$.
So in particular, during the process we can indeed apply Lemma~\ref{lem:conditioning-levels}
as needed. We showed that the following holds. \begin{proposition}
For every $\tilde{q}\in\{1,\ldots,s\}$, every machine is $\tilde{q}$-stable
in $\pseudo^{\textnormal{f}}$. 
\end{proposition}

\paragraph{Phase 2: Integral assignment for $J_{\largo}$.}

\noindent 
Based on $\pseudo^{f}$ we define an integral assignment of the long
jobs. Note that for each machine $i$ and each value $q\in[s]$, machine
$i$ is $q$-stable; we define $b_{iq}:=\sum_{j\in J_{q}}\pseudo^{f}(x_{ij})$
which is the number of jobs of $J_{q}$ that are assigned to $i$
by $\pseudo^{f}$. Since $\pseudo^{f}$ yields a valid solution to
$\assign(B,T)$, for each $q\in[s]$ we have that $\sum_{i\in [m]}b_{iq}=|J_{q}|$.
For each $q\in[s]$ we assign now the jobs in $J_{q}$ to the machines
such that each machine $i$ receives exactly $b_{iq}$ jobs from $J_{q}$.
Intuitively, all jobs in $J_{q}$ have essentially the same length
(up to a factor $1+\epsilon$), and therefore it is not relevant which
exact jobs from $J_{q}$ we assign to $i$, as long as we assign $b_{iq}$
jobs in total. Afterwards, we the short jobs in a standard greedy
list scheduling procedure: We consider the jobs in an arbitrary
order and assign each job on a machine that currently has the minimum
load among all machines. Now we are ready to prove Theorem~\ref{thm:int-gap-SA}
by showing that the load of every machine is at most $(1+\varepsilon)T$.

\noindent 

\begin{proof}[Theorem~\ref{thm:int-gap-SA}] 
 Let $\left\{ \bar{x}_{i,j}\right\} _{i\in [m],j\in J}$ denote the
computed integral assignment of the jobs to the machines, i.e., $\bar{x}_{i,j}=1$
if we assigned job $j$ on machine $i$ and $\bar{x}_{i,j}=0$ otherwise.
We first check that for each machine $i\in [m]$, we have that $\sum_{q=1}^{s}\sum_{j\in J_{q}}\bar{x}_{i,j}p_{j}\le(1+\varepsilon)T.$
Since the solution given by $\pseudo^{f}$ feasible for $\assign(B,T)$,
for each machine $i$ 
we have that $\sum_{q=1}^{s}b_{iq}\left(\frac{1}{\varepsilon}+q-1\right)\varepsilon^{2}T\le T$.
This implies for each machine $i$ that
\begin{eqnarray*}
\sum_{q=1}^{s}\sum_{j\in J_{q}}\bar{x}_{i,j}p_{j}%
 & \le & (1+\varepsilon)\sum_{q=1}^{s}\sum_{j\in J_{q}}\bar{x}_{i,j}\left(\frac{1}{\varepsilon}+q-1\right)\varepsilon^{2}T\\
 & \le & (1+\varepsilon)\sum_{q=1}^{s}\left(\frac{1}{\varepsilon}+q-1\right)\varepsilon^{2}T\sum_{j\in J_{q}}\bar{x}_{i,j}\\
 & \le & (1+\varepsilon)\sum_{q=1}^{s}\left(\frac{1}{\varepsilon}+q-1\right)\varepsilon^{2}T\cdot b_{iq}\le(1+\varepsilon)T.
\end{eqnarray*}
%
It remains to argue about the short jobs. If the global makespan does
not increase while assigning them greedily, the overall makespan remains
at most $(1+\varepsilon)T$. Otherwise, the makespan of any two machines
differ by at most $\varepsilon T$. Since $\sum_{j}p_{j}\le mT$ 
we conclude that the makespan is at most $(1+\varepsilon)T$. 
\end{proof} 

\subsubsection{\label{subsec:pf-lemma-cond} Stable conditionings: Proof of Lemma~\ref{lem:conditioning-levels}}

Recall that $\pseudo$ is a degree-$r$ pseudoexpectation with $r\ge1/\varepsilon^{2}$.
We use the following strategy to prove Lemma~\ref{lem:conditioning-levels}.
First, we identify the rightmost machine $i$ such that according
to $\pseudo$ with non-zero probability there are $1/\epsilon$ jobs
from $J_{q+1}$ assigned to $i$. Let $i_{0}$ be this machine and
let $A\subseteq [m]\times J$ denote the corresponding pairs $(i,j)$
with $i=i_{0}$ and $j\in J_{q+1}$. We apply Lemma~\ref{lem:SA-properties}
on $A$. We argue that in the resulting pseudo-expectation $\pseudo_{A}$
with non-zero probability there are $1/\epsilon$ jobs from $J_{q+1}$
assigned to $i_{L}$, let $A'\subseteq [m]\times J$ denote the corresponding
pairs. We apply Lemma~\ref{lem:SA-properties} on $A'$ as well.
In the resulting pseudoexpectation $\pseudo_{A\cup A'}$, the symmetry
breaking constraints in $\ass(B,T)$ ensure that \emph{each} machine
$i'$ between $i_{L}$ and $i_{0}$ has \emph{exactly} $1/\epsilon$
jobs from $J_{q+1}$ assigned to $i$ and therefore $i'$ is $(q+1)$-stable.
Also, no machine between $i_{0}$ and $i_{R}$ will ever get $1/\epsilon$
jobs from $J_{q+1}$ assigned to it with non-zero probability, no
matter on which sets $A''$ we might condition later. We continue
inductively: on the machines between $i_{0}$ and $i_{R}$ we look
for the rightmost machine $i$ such that according to $\pseudo_{A\cup A'}$
with some non-zero probability there are $1/\epsilon-1$ jobs from
$J_{q+1}$ assigned to $i$, etc. There are at most $1/\epsilon$
iterations in total and in each step the degree of the pseudo-expectation
decreases by at most $2/\epsilon$. Therefore, at the end, we obtain
a degree-$(r-2/\varepsilon^{2})$ pseudo-expectation in which all
machines are $(q+1)$-stable.

Now we describe our argumentation in detail. First assume that there
is no machine $i\in\{i_{L},...,i_{R}\}$ for which there exists a
set $A\subseteq [m]\times J$ with $\pseudo(x_{A})>0$, $|A|=1/\varepsilon$,
and where each tuple $(h,\upsilon)\in A$ satisfies that $h=i$ and
$\upsilon\in J_{q+1}$. In this case we define $\pseudo^{(0)}=\pseudo$
and $i_{0}=i_{L}-1$. Intuitively, in this case in our final assignment
there will be no machine in $\{i_{L},...,i_{R}\}$ that has $1/\varepsilon$
jobs from $J_{q+1}$ assigned to it. We will use later in our induction
that $\pseudo^{(0)}$ is a degree-$(r-2/\varepsilon)$ pseudoexpectation.
Otherwise let $i_{0}$ be the rightmost machine in $\{i_{L},...,i_{R}\}$,
i.e., the machine with largest index, satisfying the above for some
set $A$. We condition on $A$ and obtain the degree-$(r-1/\varepsilon)$
$\SA$ conditioning $\pseudo_{A}$. 
Recall that by Lemma~\ref{lem:SA-properties}$(a)$ each job $\upsilon\in J_{\largo}$
with $(i_{0},\upsilon)\in A$ is scheduled integrally to $i_{0}$.

\begin{lemma} \label{lem:condition-machine-1} There exists a set
$A'\subseteq [m]\times J$ with $\pseudo_{A}(x_{A'})>0$, $|A'|=1/\varepsilon$,
and for every $(h,\upsilon)\in A'$ we have that $h=i_{L}$ and $\upsilon\in J_{q+1}$.
\end{lemma}

\begin{proof} Assume that this is not the case. Then let $A'$ denote
the set of maximum size such that $\pseudo_{A}(x_{A'})>0$ and such
that each $(h,\upsilon)\in A'$ satisfies that $h=i_{L}$ and $\upsilon\in J_{q+1}$.
Observe that $|A'|\le1/\varepsilon$ by Lemma~\ref{lem:SA-properties}$(b)$,
and let $\pseudo_{A\cup A'}$ be the degree-$(r-2/\varepsilon)$ $\SA$
conditioning. 
Then, for each job $j\in J_{q}$ it holds that $\pseudo_{A\cup A'}(x_{i_{L}j})=0$,
otherwise $0<\pseudo_{A\cup A'}(x_{i_{L}j})=\pseudo_{A}(x_{A'}x_{i_{L}j})/\pseudo_{A}(x_{A'})$
and then $\pseudo_{A}(x_{A'}x_{i_{L}j})=\pseudo_{A}(x_{A'\cup\{(i_{L},j)\}})>0$,
which contradicts the maximality of $A'$. 
But then the fractional schedule given by $\pseudo_{A'}(x_{ij})$
for every $(i,j)\in [m]\times J_{\largo}$ violates the symmetry breaking
constraints of $\ass(J_{\largo},B,T)$, which is a contradiction.
 \end{proof}

Starting from $\pseudo_{A}$ we condition on $A'$ (i.e., we apply
Lemma~\ref{lem:SA-properties}) given by Lemma~\ref{lem:condition-machine-1},
obtaining $\pseudo^{(0)}=\pseudo_{A\cup A'}$. As a result, machine
$i_{L}$ and machine $i_{0}$ have both exactly $1/\varepsilon$
jobs from $J_{q}$ assigned to it. Due to the symmetry breaking constraints
of the program $\ass(J_{\largo},B,T)$ this implies that each machine
in $\{i_{L},...,i_{0}\}$ has exactly $1/\varepsilon$ jobs from $J_{q}$
(fractionally) assigned to it in $\pseudo^{(0)}$. \begin{lemma}
For each machine $i\in\{i_{L},...,i_{0}\}$ we have $\sum_{j\in J_{q}}\pseudo^{(0)}(x_{ij})=1/\varepsilon$.
\end{lemma}

\begin{proof} Machine $i_{0}$ has $1/\varepsilon$ jobs from $J_{q}$
assigned to it integrally. If there was yet another job $j\in J_{q}$
fractionally assigned to $i_{0}$ then we could condition on $(i_{0},j)$
and obtain a 
degree-$(r-1/\varepsilon-1)$ pseudoexpectation, with at least $1/\varepsilon+1$
long jobs assigned to $i_{0}$, which is a contradiction. The same
argument holds for machine $i_{L}$. The claim for the machines $\{i_{L}+1,...,i_{0}-1\}$
follows by the symmetry breaking constraints in $\ass(J_{\largo},B,T)$
that enforce the lexicographic ordering over the machines.
 \end{proof}

\begin{proof}[Lemma~\ref{lem:conditioning-levels}] Assume
by induction that for some $k\in\{0,...,1/\varepsilon-1\}$ we obtained
a degree-$(r-\sum_{\ell=1}^{k}2(1/\varepsilon+1-\ell))$ pseudoexpectation
$\pseudo^{(k)}$ such that there are machines $i_{0},...,i_{k}\in [m]$
such that for each $\ell\in\{0,...,k\}$ we have that $i_{\ell}\le i_{\ell+1}$
and each machine $w\in\{i_{\ell-1}+1,...,i_{\ell}\}$ satisfies $\sum_{j\in J_{q}}\pseudo^{(k)}(x_{wj})=1/\varepsilon-\ell$,
with $i_{-1}=i_{L}-1$ for convenience. Moreover, assume by induction
that there is no set $\Gamma\subseteq [m]\times J_{\largo}$ with $\pseudo^{(k)}(x_{\Gamma})>0$
and $|\Gamma|=1/\varepsilon-k$ such that each $(w,j)\in\Gamma$ satisfies
that $w=i_{k}+1$ and $j\in J_{q}$. The solution $\pseudo^{(0)}$
constructed above satisfies the base case of $k=0$. \\

\noindent \textit{Inductive step.} Given the solution $\pseudo^{(k)}$
we construct a solution $\pseudo^{(k+1)}$ as follows. Let $i_{k+1}$
denote the rightmost machine larger than $i_{k}$ such that there
is a set $A_{k}\subseteq [m]\times J_{\largo}$ with $\pseudo^{(k)}(x_{A_{k}})>0$,
$|A_{k}|=1/\varepsilon-k-1$, and each tuple $(w,\upsilon)\in A_{k}$
satisfies that $w=i_{k+1}$ and $\upsilon\in J_{q}$. If there is
no such machine then we define $i_{k+1}=i_{k}$ and set $\pseudo^{(k+1)}=\pseudo^{(k)}$
which is a pseudoexpectation of degree $r-\sum_{\ell=1}^{k+1}2(1/\varepsilon+1-\ell)$.
Otherwise we condition on $A_{k}$ and obtain $\pseudo_{A_{k}}^{(k)}$,
of degree $r-\sum_{\ell=1}^{k}2(1/\varepsilon+1-\ell))-(1/\varepsilon-k-1)$.
Following the same lines of Lemma~\ref{lem:condition-machine-1}
we have that there exists a set $B_{k}\subseteq [m]\times J_{\largo}$
with $\pseudo_{A_{k}}^{(k)}(x_{B_{k}})>0$, $|B_{k}|=1/\varepsilon-k-1$,
and each tuple $(h,\upsilon)\in B_{k}$ satisfies that $h=i_{k}+1$
and $\upsilon\in J_{q}$. 
We then define $\pseudo^{(k+1)}=\pseudo_{A_{k}\cup B_{k}}^{(k)}$,
which is a pseudoexpectation of degree $r-\sum_{\ell=1}^{k+1}2(1/\varepsilon+1-\ell)$.
\begin{claim} \label{claim:induction} For each $w\in\{i_{k}+1,...,i_{k+1}\}$
we have that $\sum_{j\in J_{q}}\pseudo^{(k+1)}(x_{wj})=1/\varepsilon-k-1$.
\end{claim} 

We see how to conclude the lemma and then we show check the claim.
Since we chose machine $i_{k+1}$ to be the rightmost machine with
the claimed properties, 
$\pseudo^{(k+1)}$ satisfies the induction hypothesis for $k+1$.
Finally, we define $\pseudo^{\mathsf{stab}}=\pseudo^{(1/\varepsilon)}$,
which yields that $\pseudo^{\mathsf{stab}}$ is a pseudoexpectation
of degree $r-2/\varepsilon^{2}$, since $\sum_{\ell=1}^{1/\varepsilon}2(1/\varepsilon+1-\ell)\le2~/\varepsilon^{2}.$
That concludes the lemma.

\begin{proof}[Claim~\ref{claim:induction}]
On machine $w=i_{k}+1$ we conditioned
on the set $B_{k}$ due to the previous claim with $|B_{k}|=1/\varepsilon-k-1$.
Therefore $\sum_{j\in J_{q}}\pseudo^{(k+1)}(x_{wj})\ge1/\varepsilon-k-1$.
On the other hand, if $\sum_{j\in J_{q}}\pseudo^{(k+1)}(x_{wj})>1/\varepsilon-k-1$
then there must be a pair $(w,j)\notin B_{k}$ with $\pseudo^{(k+1)}(x_{wj})>0$
and $j\in J_{q}$. But this implies that $\pseudo^{(k)}(x_{B_{k}\cup\{(w,j)\}})>0$
which contradicts the induction hypothesis. With the same reasoning
we argue that $\sum_{j\in J_{q}}\pseudo^{(k+1)}(x_{i_{k+1}j})=1/\varepsilon-k-1$.
The claim for the machines in $\{i_{k}+2,...,i_{k+1}-1\}$ then follows
from the symmetry breaking constraints in $\ass(J_{\largo},B,T)$
that enforce the lexicographic over the machines.  
\end{proof}
 %
\end{proof}


\subsection{A faster LP based approximation scheme}

\label{sec:faster-ptas}

In Section~\ref{sec:Integrality-gap-SA} we proved that after applying~$\nrounds$
rounds of Sherali-Adams to $\assign(B,T)$ we obtain a linear relaxation
with an integrality gap of at most $1+\varepsilon$. In this section,
we add to $\assign(B,T)$ a set of constraints that we refer to
as the \emph{ordering constraints}, obtaining a linear program that we refer
to as $\order(B,T)$. Intutitively, we prove that if we apply only
$\mathrm{poly(1/\varepsilon)}$ rounds of Sherali-Adams to this new
program then its integrality gap drops to $1+\varepsilon$. On the
other hand, it might be that there is no optimal solution (i.e.,
a solution with makespan $\mathrm{OPT}$) that satisfies the ordering
constraints and in particular it might be that $\order(B,\mathrm{OPT})$
does not have a feasible solution (in contrast to $\assign(B,\mathrm{OPT})$
which is always feasible). However, we can guarantee that
$\order(B,(1+\varepsilon)\mathrm{OPT})$ is always feasible.\\

\noindent{\it Ordering Constraints.}
Roughly speaking, we use a new set of constraints that allow
us to break symmetries due to permutations of jobs in the same class
$J_{q}$ (and not only the symmetries corresponding to permutations
of machines), which is a key difference to the approach used for the
approximation scheme of Section~\ref{sec:Integrality-gap-SA}. For
each $q\in\{1,\ldots,s\}$ assume that $J_{q}=\{j_{q,1},j_{q,2},...,j_{q,|J_{q}|}\}$
and we impose that the jobs in $J_{q}$ are scheduled in this order,
i.e., if jobs $j_{q,\ell}$ and $j_{q,\ell+1}$ are scheduled on machines
$i\in [m]$ and $h\in [m]$ for some $\ell\in\{1,\ldots,|J_{q}|-1\}$, then $i\le h$.
To enforce this, for each $q\in\{1,\ldots,s\}$, each $\ell\in\{1,...,|J_{q}|-1\}$,
and for each $h\in [m]$ we add to $\assign(B,T)$ the constraint 
\begin{equation}
\sum_{i=1}^{h}x_{ij_{q,\ell}}\ge\sum_{i=1}^{h}x_{ij_{q,\ell+1}}.\label{eq:order-constraints}
\end{equation}
Denote by $\order(B,T)$ the LP obtained by adding the above set of
constraints to $\assign(B,T)$. It might be that there is no feasible
solution to $\order(B,\mathrm{OPT})$. However, in the following lemma
we show that there exists always a solution to $\order(B,(1+\varepsilon)\mathrm{OPT})$.
\begin{lemma} There exists a feasible integral solution to $\order(B,(1+\varepsilon)\mathrm{OPT})$.
\end{lemma}

\begin{proof} Consider the integral vector $x$ to $\assign(B,\mathrm{OPT})$
that stems from the optimal solution to the given instance. For each
machine $i\in [m]$ denote by $\conf_{i}(x)$ its vector $(a_{i,1},...,a_{i,s})$
as defined in Section~\ref{sec:Integrality-gap-SA}. 
For each $q\in\{1,\ldots,s\}$ we rearrange the jobs in $J_{q}$ on
the machines such that the resulting schedule satisfies the order
constraints and on each machine the number of jobs from each set $J_{q}$
stays the same. Given $q\in\{1,\ldots,s\}$, for each machine $i\in [m]$
let $b_{iq}$ denote the number of jobs from $J_{q}$ on $i\in [m]$
in a schedule with optimal makespan $\mathrm{OPT}$. In our new schedule,
for each machine $i\in [m]$ we define $c_{iq}=\sum_{h=1}^{i}b_{hq}$
and we set $c_{0q}=0$. Then we assign to each machine $i\in [m]$ the
jobs $\{j_{q,c_{i-1,q}+1},...,j_{q,c_{iq}}\}$. We repeat this operation
for every $q\in\{1,\ldots,s\}$. Within each set $J_{q}$ the processing
times of two jobs can differ by at most $\varepsilon^{2}\mathrm{OPT}$.
Each machine has at most $1/\varepsilon$ long jobs assigned to it.
Therefore, due to our reassignment of jobs the makespan of the schedule
can increase by at most $\frac{1}{\varepsilon}\cdot\varepsilon^{2}\cdot \mathrm{OPT}=\varepsilon\cdot \mathrm{OPT}$
on each machine. This integral schedule yields a solution to $\order(B,(1+\varepsilon)\mathrm{OPT})$.
 \end{proof} In the remainder of this section we prove the following
theorem. 

\begin{theorem} \label{thm:int-gap-order-constraints} Consider
a value $T>0$ and suppose there exists a degree $\pnrounds$ SA
pseudoexpectation for $\order(B,T)$. Then, there exists an integral
solution for $\order(B,(1+\varepsilon)T)$ and it can be computed
in polynomial time. 
\end{theorem}

As before we first construct a solution for the long jobs only and
afterwards argue that we can add the short jobs with only marginal
increase of the makespan. For a degree-$r$ pseudoexpectation $\pseudo$
and a set of machines $M^{*}$, we say that $M^{*}$ is \emph{focused
}if each job $j\in\Jl$ is either completely assigned to machines
in $M^{*}$, i.e., $\sum_{i\in M^{*}}\pseudo(x_{ij})=1$, or to no
machine in $M^{*}$, i.e., $\sum_{i\in M^{*}}\pseudo(x_{ij})=0$ and
the same holds for any conditioning obtained from $\pseudo$.\\

\noindent{\it Overview.} We first apply Lemma~\ref{lem:conditioning-levels}
to make sure that each machine is $1$-stable. This partitions the
machines into sets $\{M_{0},\ldots,M_{1/\epsilon}\}$ such that
the machines in each set $M_{\ell}$ have exactly $\ell$ jobs from
$J_{1}$. The machines in each set $M_{\ell}$ are consecutive. Then,
intuitively, for each set $M_{\ell}$ we take the rightmost machine
$i$ and condition on every single long job on $i$ (so not just on the jobs in $J_1$). As a result,
due to the ordering constraints each set $M_{\ell}$ is focused. This
operation is formalized in Lemma~\ref{lem:conditioning-levels-order}.
Then, we observe that for each set of machines $M_{\ell}$ we obtain
a degree-$(r-4/\varepsilon^{3})$ pseudoexpectation for $\order(\Jll,B,T,M_{\ell})$
for some set $\Jll\subseteq\Jl$ that is completely independent of
all other sets $M_{\ell'}$ with $\ell'\ne\ell$. Therefore, we can
recurse on each set of machines $M_{\ell}$ \emph{independently} such
that the degree of our pseudoexpectation drops by at most $4/\varepsilon^{3}$
in each level. Since there are only $1/\epsilon^{2}$ levels, it suffices
to start with a pseudoexpectation of degree at most $4/\epsilon^{5}$.

\begin{lemma} \label{lem:conditioning-levels-order} Consider $q\in\{0,1,\ldots,s\}$,
integers $a_{1},...,a_{q}$ and a degree-$r$ pseudoexpectation $\pseudo$
such that each machine $i\in [m]$ is $(\tilde{q},a_{\tilde{q}})$-stable
for each $\tilde{q}\in\{1,...,q\}$. Then there is a degree $r-4/\varepsilon^{3}$
pseudoexpectation $\pseudo^{\text{focus}}$ 
obtained from $\pseudo$ via conditioning on at most $4/\varepsilon^{3}$
variables such that each machine in $[m]$ is $\hat{q}$-stable for
each $\hat{q}\in\{1,...,q+1\}$ and there is a partioning of $[m]$
given by $\{M_{1},...,M_{k}\}$ of consecutive machines such that
for each $\ell\in\{1,\ldots,k\}$ we have that $M_{\ell}$ is focused
and each machine $i\in M_{\ell}$ is $(\tilde{q}+1,\ell)$-stable.
\end{lemma}

\begin{proof} First, we apply Lemma~\ref{lem:conditioning-levels}
with the same value $q$, the integers $a_{1},...,a_{q}$, 
and with $M^{*}=[m]$. Let $\pseudo^{\text{stab}}$ the degree $(r-2/\varepsilon^{2})$
pseudoexpectation obtained. 
Since $\pseudo^{\text{stab}}$ is $(q+1)$-stable, we obtain a partition
of $[m]$ given by $\{M_{0},...,M_{k}\}$ with 
$k\le1/\varepsilon$ such that each machine $i\in M_{\ell}$
is $(q+1,\ell)$-stable for each $\ell\in\{0,...,k\}$. For
$\ell\in\{0,\ldots,k\}$ we proceed as follows: For each $\tilde{q}\in\{1,\ldots,s\}$
let $\eta(\tilde{q},\ell)$ be the largest index such that there is
a job $j_{\tilde{q},\eta(\tilde{q},\ell)}$ with 
\[\pseudo^{\text{stab}}(x_{ij_{\tilde{q},\eta(\tilde{q},\ell)}})>0 \text{ for some machine $i\in M_{\ell}$}.\] 
We condition on $(i,j_{\tilde{q},\ell(\tilde{q},\ell)})$.
More precisely, we iterate over the values $\tilde{q}\in\{1,\ldots,s\}$
and condition on the respective jobs one by one, obtaining a pseudoexpectation
$\pseudo^{\text{focus}}$. 
Since $2/\varepsilon^{2}-s(1/\varepsilon+1)\le4/\varepsilon^{3}$
this pseudoexpectation is of degree $r-4/\varepsilon^{3}$.
We claim that in $\pseudo^{\text{focus}}$ each subset $M_{\ell}$
with $\ell\in\{0,\ldots,k\}$ is focused. At the beginning the
set $[m]$ is focused. At the first step $\ell=1$, for each $\tilde{q}\in\{1,\ldots,s\}$
either there is no job $j\in J_{\tilde{q}}$ fractionally assigned
on a machine in $M_{1}$ or we conditioned on the job $j_{\tilde{q},\eta(\tilde{q},1)}$
with largest index $\eta(\tilde{q},1)$. Hence, due to the order constraints,
no job $j_{\tilde{q},\eta'}$ with $\eta'>\eta(\tilde{q},1)$ can
be fractionally assigned to a machine in $M_{1}$. Hence, for each
$\tilde{q}\in\{1,\ldots,s\}$ there is a set $\tilde{J}_{\tilde{q}}\subseteq J_{\tilde{q}}$
such that all jobs in $\tilde{J}_{\tilde{q}}$ are assigned on $M_{1}$
and no job in $J_{\tilde{q}}\setminus\tilde{J}_{\tilde{q}}$ is fractionally
assigned. Hence, $M_{1}$ is focused and $[m]\setminus M_{1}$ is also
focused. The remainder follows by induction with the same argument.
 \end{proof}

\noindent{\it Algorithm.} In the remaining fix $r=\pnrounds$. We take a degree-$r$-pseudoexpectation for $\order(B,T)$.
We first apply Lemma~\ref{lem:conditioning-levels-order}
with $q=0$ and obtain a solution $\pseudo^{\text{focus}}$. For each
group $M_{\ell}$ with $\ell\in\{1,\ldots,k\}$ denote by $\Jll$
the jobs from $\Jl$ assigned on $M_{\ell}$ in according to $\pseudo^{\text{focus}}$.
Then, for each $\ell\in\{1,\ldots,k\}$ this yields a pseudoexpectation
$\pseudo_{\ell}^{\text{focus}}$ for the program $\order(\Jll,B,T,M_{\ell})$
of degree $r-4/\varepsilon^{3}$, 
in which each machine is $1$-stable. Intuitively, we continue recursively
on each part. The depth of this recursion is $s$. In each level,
we condition on at most $4/\varepsilon^{3}$ variables. Hence,
if we obtain a pseudoexpectation of degree $s\cdot4/\varepsilon^{3}\le\pnrounds$
in $\order(B,T)$ then we obtain that there exists a solution for
$\order(B,T)$ that is $q$-stable for each $q\in\{1,\ldots,s\}$.
Formally, we prove the following lemma by induction. 
\begin{lemma} \label{lem:conditioning-levels-order-induction}
Consider $q\in\{0,1,\ldots,s\}$, integers $a_{1},...,a_{q}$ and
a degree $4(s-q)/\varepsilon^{3}$ pseudoexpectation 
such that each machine $i\in [m]$ is $(\tilde{q},a_{\tilde{q}})$-stable
for each $\tilde{q}\in\{1,...,q\}$. Then, there is a solution in
$\order(B,T)$ such that each machine $i\in [m]$ is $\hat{q}$-stable
for each $\hat{q}\in\{1,\ldots,s\}$. \end{lemma}

\begin{proof} We prove the lemma by induction. If $q=s$ then the
lemma is trivially true. Now suppose that the lemma is true for some
value $q+1$. 
Given a pseudoexpectation corresponding to solution to $\order(B,T)$
we apply Lemma~\ref{lem:conditioning-levels-order}
and obtain a solution $\pseudo^{\text{focus}}$ and the partition
$\{M_{1},\ldots,M_{k}\}$. For each $\ell\in\{1,\ldots,k\}$ this
yields a pseudoexpectation $\pseudo_{\ell}^{\text{focus}}$ with degree
$4(s-q-1)/\varepsilon^{3}$ 
such that in $\pseudo_{\ell}^{\text{focus}}$ each machine $i\in M_{\ell}$
is $(q+1,\ell)$-stable and also $\tilde{q}$-stable for each $\tilde{q}\in\{1,...,q\}$.
On each pseudoexpectation $\pseudo_{\ell}^{\text{focus}}$ with $\ell\in\{1,\ldots,k\}$
we apply the induction hypothesis and obtain a solution $\text{\ensuremath{x^{\ell}}}\in\order(\Jll,B,T,M_{\ell})$
such that in $\ensuremath{x^{\ell}}$ each machine $i\in [m]$ is $\hat{q}$-stable
for each $\hat{q}\in\{1,\ldots,s\}$. We define the solution to be
the direct sum of the solutions $x^{\ell}$ over $\ell\in\{1,\ldots,k\}$.
 \end{proof} 
The lemma above
yields that if there exists degree $4/\varepsilon^{5}$ pseudoexpectation to $\order(B,T)$
then there exists a solution $x\in\order(B,T)$ in which each machine
is $\hat{q}$-stable for each $\hat{q}\in\{1,\ldots,s\}$. The 
assignment of the long jobs to the machines is identical to the proof of Lemma~\ref{lem:conditioning-levels}.
Finally, we add the short jobs greedily like in Section~\ref{sec:Integrality-gap-SA}.
This completes the proof of Theorem~\ref{thm:int-gap-order-constraints}.
\bibliography{bib1}{}
\bibliographystyle{abbrv}

\appendix
\section{Proof of Theorem 4}
\label{sec:appendixA}

We show how to prove Theorem~\ref{thm:GP} following the lines in the work of Raymond et al.~\cite{RSST18}.
We need a few intermediate results, and the symmetry reduction theorem from Gaterman \& Parrilo~\cite{GP04}, stated in our setting.

\begin{theorem}[\hspace{-0.01cm}\cite{GP04}]
\label{thm:GPoriginal}
Suppose that $g\in \RR[y]/\sched$ is a degree-$\ell$ $\sos$ and $S_m$-invariant polynomial.
For each partition $\lambda\vdash m$, let $\tau_{\lambda}$ be a tableau of shape $\lambda$ and let $\{b_1^{\lambda},\ldots,b_{m_{\lambda}}^{\lambda}\}$ be a basis $\rowspace$.
Then, for each partition $\lambda\vdash m$ there exists a $m_{\lambda}\times m_{\lambda}$ positive semidefinite matrix $Q_{\lambda}$ such that 
$g=\sum_{\lambda \vdash m}\langle Q_{\lambda},Y^{\lambda}\rangle$, where $Y^{\lambda}_{ij}=\sym(b_i^{\lambda}b_j^{\lambda})$.
\end{theorem}

Given two partitions $\lambda,\mu$, we say that $\lambda \unrhd \mu$ if $\lambda\ge_{\lex} \mu$ and the number of parts of $\mu$ is at least the number of parts of $\lambda$.
The following lemma is a variant of \cite[Theorem 2]{RSST18} for the action of the symmetric group in our setting.
Together with the theorem of Gatermann \& Parrilo this yields Theorem~\ref{thm:GP}.

\begin{lemma}
\label{lem:dim-zero}
The dimension $m_{\lambda}$ of ${\bf{Q}}^{\ell}_{\lambda}$ in the isotypic decomposition of $\bf{Q}^{\ell}$ is zero unless $\lambda\ge _{\lex} (m-\ell,1,\ldots,1)$.
\end{lemma}

\begin{proof}
Let $y_S$ be a monomial of degree at most $\ell$ with $S=\{(i_k,C_k):k\in [\ell]\}$.
In particular, $|\{i_k:k\in [\ell]\}|\le \ell$.
Let $\tau$ be any tableau with shape $(m-\ell,1,\ldots,1)$, where the tail of $\tau$ contains every elements of $\{i_k:k\in [\ell]\}$.
The subgroup $\mathcal{R}_{\tau}$ fixes $S$, therefore $y_S\in {\bf{W}}^{\ell}_{\tau}$, and we have then 
\[{\bf{Q}}^{\ell}\subseteq \bigoplus_{\tau:\text{shape}(\tau)=(m-\ell,1,\ldots,1)}{\bf{W}}^{\ell}_{\tau}\subseteq  \bigoplus_{\lambda \unrhd (m-\ell,1,\ldots,1)}{\bf{Q}}^{\ell}_{\lambda},\]
where the second containment holds by~\cite[Lemma 1]{RSST18}.
To conclude, observe that if $\lambda \unrhd (m-\ell,1,\ldots,1)$ then $\lambda_1\ge m-\ell$.
Since $\lambda\vdash m$, the maximum number of parts for $\lambda$ is $m-\lambda_1\le \ell$, that is, $\lambda$ has at most $\ell+1$ parts. 
Therefore, $\lambda \unrhd (m-\ell,1,\ldots,1)$ if and only if $\lambda \ge_{\lex} (m-\ell,1,\ldots,1)$.
\end{proof}

\begin{proof}[Proof of Theorem~\ref{thm:GP}]
Let $g\in \RR[y]/\sched$ be a degree-$\ell$ $\sos$ and $S_m$-invariant polynomial.
By Theorem~\ref{thm:GPoriginal} and Lemma~\ref{lem:dim-zero}, for each $\lambda\in \Lambda_{\ell}$ there exists a positive semidefinite matrix $Y^{\lambda}$ such that 
$g=\sum_{\lambda \in \Lambda_{\ell}}\langle Q_{\lambda},Y^{\lambda}\rangle$.
Since $\{b_1^{\lambda},\ldots,b_{m_{\lambda}}^{\lambda}\}\subseteq \text{span}(\mathcal{P}^{\lambda})$, there exists a real matrix $\mathcal{T}_{\lambda}$ such that $\mathcal{T}_{\lambda}(p_1^{\lambda},\ldots,p_{\ell_{\lambda}}^{\lambda})=(b_1^{\lambda},\ldots,b_{m_{\lambda}}^{\lambda})$.
Consider the congruent transformation $M_{\lambda}=\mathcal{T}_{\lambda}^{\top}Q_{\lambda}\mathcal{T}_{\lambda}$.
In particular, $M_{\lambda}$ is also positive semidefinite.
Furthermore, 
\[{\bf b}^{\top}Q_{\lambda}{\bf b}=(\mathcal{T}_{\lambda}{\bf p})^{\top}Q_{\lambda}(\mathcal{T}_{\lambda}{\bf p})={\bf p}^{\top}M_{\lambda}{\bf p},\]
where ${\bf b}=(b_1^{\lambda},\ldots,b_{m_{\lambda}}^{\lambda})$ and ${\bf p}=(p_1^{\lambda},\ldots,p_{\ell_{\lambda}}^{\lambda})$.
That is, $g=\sum_{\lambda \in \Lambda_{\ell}}\langle Q_{\lambda},Y^{\lambda}\rangle=\sum_{\lambda \in \Lambda_{\ell}}\langle M_{\lambda},Z^{\lambda}\rangle$.
\end{proof}

\section{\sos{} Lower Bound for the Assignment Linear Program}
\label{sec:appendixB}

We now show that the lower bound of Theorem~\ref{thm:negative} translates to the assignment linear program.
Recall that the $r$-th level of the \sos{} hierarchy corresponds to a semidefinite program with variables $y_S$ for any subset $S\subseteq E$ with $|S|\le r$. The inequalities defining this program can be obtained by considering properties \ref{prop:normal}-\ref{prop:SA} in the definition of degree-$r$ \sos{} pseudoexpectations and identifying $\widetilde{\mathbb{E}}(x_S)=y_S$; see for example~\cite{mastrolilli2017high} for details. For any polytope $P\subseteq [0,1]^E$, we denote by $\sos{}_r(P)$ the projection of the $r$-th level of the \sos{} hierarchy over $y_i=y_{\{i\}}$ for each $i\in E$.
Au and Tun\c{c}el~\cite[Proposition 1]{au_elementary_2018} showed that for any polytope $P\subseteq [0,1]^E$, if $L:\mathbb{R}^E\rightarrow \mathbb{R}^E$ is an affine transformation such that $L(x)\in[0,1]^E$ for all elements in the unit hypercube $x\in[0,1]^E$, then $\sos{}_r(L(P))= L(\sos{}_r(P))$. In our case, we consider the configuration linear program and the assignment linear program within the same space. Let $T$ be a target makespan and consider 
 \[
  P = \mathbb{R}^{[m]\times J}\times\clp{}(T)=  \{(x,y)\in \mathbb{R}^{[m]\times J}\times\mathbb{R}^{[m]\times\mathcal{C}}: y \in \clp{}(T)\}.
 \]
 We define the projection $L(x,y)=(x',0)$ where $x'$ is defined as
\[x'_{ij} = \frac{1}{n_{p_j}}\sum_{C\in \mathcal{C}} m(C,p_j)\cdot y_{iC}  \qquad \text{for all } i\in [m]\text{ and for all }j\in J.\] Notice that $x'$ belongs to the assignment linear program, and hence $L(P)\subseteq \text{assign}(T)\times [0,1]^{[m]\times\mathcal{C}}$ is within the unit hypercube and the result by Au and Tun\c{c}el can be applied. 
Therefore,  \[L(\sos{}_{r+1}(P))= \sos{}_{r+1}(L(P))\subseteq \sos{}_r(\text{assign}(T)\times [0,1]^{[m]\times\mathcal{C}}),\] where the last inclusion follows since $L(P)\subseteq \text{assign}(T)\times [0,1]^{[m]\times\mathcal{C}}$ and the general property of the next lemma.
We remark that this is enough to get an integrality gap of $1.0009$ for $\Omega(n)$ rounds of the \sos{} hierarchy applied to the assignment linear program.

\begin{lemma}If $P$ and $Q$ are two polytopes with $P\subseteq Q$, then $\sos_{r+1}{}(P)\subseteq \sos_r{}(Q)$.\end{lemma}
\begin{proof}
Let us assume that $P =\{x\in \RR^n:Ax\le b\}$ and $Q=\{x\in \RR^n:Cx\le d\}$ for some $A\in\mathbb{R}^{m\times n},b\in \mathbb{R}^m$, $C\in \mathbb{R}^{p\times n}$ and $d\in\mathbb{R}^{p}$. Let $a_i^{\top}$ be the $i$-th row of $A$ and $c_i^{\top}$ the $i$-th row of $C$. We will show that a degree-$(r+1)$ \sos{} pseudoexpectation for $P$ is also a degree-$r$ \sos{} pseudoexpectation for $Q$. Indeed, recall that if $P\subseteq Q$, then every inequality $c_i^{\top} x\le d_i$, where $c_i$ is the $i$-th row of $C$, is a valid inequality for $P$. Hence, by Farkas lemma, for each row $i\in [p]$ there exists a non-negative vector $\gamma\in \mathbb{R}^m$ such that $c_i=\gamma^{\top} A$ and $\gamma^{\top} b \le d_i$. Let $\widetilde{\mathbb{E}}$ be a degree-$(r+1)$ \sos{} pseudoexpectation for $P$. We need to show that property \ref{prop:sos2} is satisfied for every inequality $(d_i-c_i^{\top}x)\ge 0$, with $i\in [p]$. Let $f\in \RR[x]/{\I}_n$ with $\deg\left(\overline{f^2(d_i-c_i^{\top}x)}\right)\le r$. By basic algebraic manipulation it holds that
\begin{equation*}
 \widetilde{\mathbb{E}}(f^2(d_i-c_i^{\top}x)) = (d_i - \gamma^{\top}b)\widetilde{\mathbb{E}}(f^2)+\sum_{j=1}^m\gamma_j \widetilde{\mathbb{E}}(f^2(b_j-a_j^{\top}x)) \ge 0,
\end{equation*} 
where the last inequality follows from the construction of $\gamma$, the fact that for each $j\in [m]$ we have $\deg\left(\overline{f^2(b_j-a_j^{\top}x)}\right)\le r+1$,  and hence $\widetilde{\mathbb{E}}(f^2(b_j-a_j^{\top}x))\ge 0$ and $\widetilde{\mathbb{E}}(f^2)\ge 0$.

\end{proof}

\end{document}